%% file: main.tex
\documentclass[conference, 10pt]{IEEEtran}
\IEEEoverridecommandlockouts
\usepackage{supertech-llncs-light}
\usepackage{booktabs} % for formal tables

\def\BibTeX{{\rm B\kern-.05em{\sc i\kern-.025em b}\kern-.08em
    T\kern-.1667em\lower.7ex\hbox{E}\kern-.125emX}}

\begin{document}

%% we may remove the following format instructions for conformation
%\special{papersize=8.5in,11in}
% \setlength{\pdfpageheight}{\paperheight}
% \setlength{\pdfpagewidth}{\paperwidth}
% \setlength{\belowdisplayskip}{0ex} \setlength{\belowdisplayshortskip}{0ex}
% \setlength{\abovedisplayskip}{0ex} \setlength{\abovedisplayshortskip}{0ex}
% % %% set the space before/after in-text figures
% \setlength{\floatsep}{0ex}
% \setlength{\intextsep}{1 \baselineskip} 
% \setlength{\textfloatsep}{1 \baselineskip}
% \setlength{\abovecaptionskip}{0pt} \setlength{\belowcaptionskip}{0ex}

\title{Nested Dataflow Algorithms for Dynamic Programming 
Recurrences with more than O(1) Dependency
\thanks{Shanghai Natural Science Funding (No. 18ZR1403100)}
}

% \subtitle{Extended Abstract}

% \author{Anonymized for Blind Peer Review}

\punt{% begin punt
\author[1]{Yuan Tang\corref{cor}\fnref{aff}}
\cortext[cor]{Corresponding Author} 
\fntext[aff]{Also affiliated with the Shanghai Key Lab. of 
Intelligent Information Processing}
\ead{yuantang@fudan.edu.cn}
\address[1]{
School of Software, School of Computer Science, Fudan University,
220 Handan Road, Shanghai 200433, P.R.China}
}% end punt

\author{\IEEEauthorblockN{Yuan Tang}
    \IEEEauthorblockA{\textit{School of Computer Science, Fudan University}\\
    Shanghai, P. R. China\\
    yuantang@fudan.edu.cn}
}
\maketitle

% \newpageafter{author}

\input{abstract}
% \newpageafter{abstract}

\maketitle

\input{intro}
\input{model}
\input{sublinear-lws}
\input{sublinear-gap}

\input{relWork}

% \input{ack}

% We recommend abbrvnat bibliography style.

\bibliographystyle{IEEEtran}
\bibliography{papers}

% \appendix
% \input{mm}
% \input{star-lws}
% \input{star-gap}
% \input{star-gap-app}
% \input{expr}

\end{document}

%% file: abstract.tex
\begin{abstract}
% State the problem
% Why it's an interesting problem
% What your solution achieves
% What follows from your solution: Impact of your solution

Dynamic programming problems have wide applications in real world
and have been studied extensively in both serial and parallel
settings. In 1994, Galil
and Park developed work-efficient and sublinear-time algorithms 
for several important dynamic programming problems based on
the closure method and matrix
product method. However, in the same paper, they raised an open
question whether such an algorithm exists for the general GAP 
problem.
In this paper, we answer their question by developing the first 
work-efficient and sublinear-time GAP
algorithm based on the closure method and Nested Dataflow method.
We also improve the time bounds of classic work-efficient,
cache-oblivious and cache-efficient algorithms for the 1D 
problem and GAP problem, respectively.
\punt{% begin punt
Though Galil and Park's original paper did not consider 
space and cache complexities, these
two requirements are important for an algorithm to achieve
good performance in practice. 
}% end punt
It remains an interesting question if we can further bound the
GAP algorithm's space and cache bounds to be asymptotically 
optimal without sacrificing work or time bounds.
\end{abstract}

\punt{% begin punt

\ccsdesc[500]{Theory of computation~Divide and conquer}
\ccsdesc[500]{Theory of computation~Dynamic programming}
\ccsdesc[500]{Computing methodologies~Shared memory algorithms}
}% end punt


%% file: intro.tex
\secput{intro}{Introduction}
%
\input{contri-table.tex}
%
% \input{symbol-table.tex}
%
% What is DP and why it's an interesting problem 
%
Dynamic Programming (DP) is a general problem-solving technique 
that solves a large problem optimally by recursively breaking
down into sub-problems and solving those sub-problems optimally
\cite{wikiDP}. If a problem can be solved this way, it is said 
to have optimal substructure.
DP has wide applications in various fields from
aerospace engineering, control theory, operation research, 
biology, economics, and computer science \cite{GalilGi89, 
GalilPa92, GalilPa94}. 

% What is the difficulty in solving a DP recurrence
%
By the nature, a problem solvable by the DP technique, is 
usually presented by a set of recurrence equations. 
\aeqreftwo{1D}{gap} are two such examples. In solving a 
set of DP recurrences, it usually boils down to update 
every entries of a multi-dimensional table by some order. 
It is this serial order that prevents an efficient parallelization.

% Current two approaches: Galil and Park's closure method and COP
%
In the literature, there are two classic ways to parallelize a 
DP algorithm. 
Based on the well-known closure method and matrix product method,
Galil
and Park \cite{GalilPa94} gave out several work-efficient and
sublinear-time algorithms for DP recurrences with more than
$O(1)$ dependency, including the 1D problem \cite{HirschbergLa87}, 
the parenthesis problem which computes the minimum cost of
parenthesizing $n$ elements \cite{Yao80}, and the secondary 
structure of RNA without multiple loops \cite{WatermanSm78}. 
They raised an open problem in their paper 
if there exists a better algorithm for the general
edit distance problem when allowing gaps of insertions and 
deletions \cite{GalilGi89} (the GAP problem). 
On the other hand, Chowdhury, Le, and Ramachandran 
\cite{ChowdhuryRa06, Chowdhury07, ChowdhuryLeRa10} developed
a set of cache-oblivious parallel (COP) and cache-efficient
algorithm for DP problems with $O(1)$ or more than $O(1)$ 
dependency.  Their algorithms are usually work-efficient
but super-linear in time. 
Dinh et al \cite{TangYoKa15, DinhSiTa16} observed that the 
classic COP method may introduce excessive control
dependency among recursively derived sub-problems, and that
excessive control dependency actually increases the time 
bound (critical-path length of its computational DAG) of 
algorithm, therefore extended the method to the 
Nested Dataflow (ND) method. 
\punt{% begin punt
However the latest COP approach still can not achieve a sublinear
time bound like that of the closure and matrix product methods 
\cite{GalilPa94}.
}% end punt
Besides the above two big classes of DP algorithms, we will 
discuss other related works in \secref{concl}.

\paragrf{Problems: }
Noticing that Galil and Park's approach does not give an 
work-efficient and sublinear-time algorithm for the general 
GAP problem
and the latest COP approach does not achieve a sublinear time
bound, we present a new framework for the parallel
computation of DP recurrences with more than $O(1)$ dependency
based on a novel combination of the closure method and ND method.
We demonstrate
the framework by working on the following two problems:

\begin{enumerate}[P1.]
    \item Given a real-valued function $\proc{w}(\cdot, \cdot)$
        , which can be computed with no memory access in $O(1)$ 
        time, and initial value $D[0]$, compute 
        \begin{align}
        D[j] &= \min_{0 \leq i < j} {D[i] + w(i, j)} & & \text{for $1 \leq j \leq n$} 
        \label{eq:1D}
        \end{align}
        This problem was called the least weight subsequence (LWS)
        problem by Hirschberg and Larmore \cite{HirschbergLa87}.
        We will call it \emph{1D} problem following the convention
        of Galil and Park \cite{GalilPa94}. 
        Its applications include, but not limited to, the optimum
        paragraph formation and finding a minimum height B-tree.
        We will study the 1D problem in \secref{sub-1D} as
        a prerequisites and 1D simplification of the more 
        complicated GAP problem.

    \item Given $w$, $w'$, $s_{ij}$, which can be computed in 
        $O(1)$ time with no memory access, and $D[0, 0] = 0$,  
        compute
        \begin{align}
        D[i, j] = \min
            \left\{
            \begin{array}{l}
            D[i - 1, j - 1] + s_{ij} \\
            \min_{0 \leq q < j} \{D[i, q] + w(q, j)\} \\
            \min_{0 \leq p < i} \{D[p, j] + w'(p, i)\} 
            \end{array}
            \right.
        \label{eq:gap}
        \end{align}
        for $0 \leq i \leq m$ and $0 \leq j \leq n$. 
        We assume that $m$ and $n$ are of the same order of
        magnitude. This is the problem of computing the edit
        distance when allowing gaps of insertions and deletions
        \cite{GalilGi89}. We will call it \emph{GAP} problem
        following the convention of Galil and Park 
        \cite{GalilPa94}.
        Its applications include, but not limited to, molecular
        biology, geology, and speech recognition.
\end{enumerate}

% List of contribution bullets
\paragrf{Our Results (\figref{contri-table}): }
\begin{enumerate}
    \item The 1D problem: 
        \begin{enumerate}
            \item We give out a linear $O(n)$ time and optimal 
                $O(n^2)$ work 1D algorithm (\thmref{nd-1D}
                in \secref{sub-1D}), which achieves optimal 
                $O(n)$ space and optimal
                $O(n^2/(BM))$ cache bounds in a cache-oblivious
                fashion. This result improves over the prior
                cache-oblivious and cache-efficient algorithm 
                on time bound.

            \item We give out a sublinear $O(\sqrt{n} \log n)$
                time and optimal $O(n^2)$ work 1D algorithm 
                (\thmref{sub-nd-1D} in \secref{sub-1D}) with
                non-optimal $O(n^2)$ space and $O(n^2/B)$
                cache bounds. This algorithm provides new 
                insight on how 
                to solve DP recurrences with more than $O(1)$ 
                dependency and is a prerequisite for the 
                more complicated GAP problem.
        \end{enumerate}

    \item The GAP problem: 
        \begin{enumerate}
            \item We give out a superlinear $O(n \log n)$ time,
                optimal $O(n^3)$ work GAP algorithm 
                (\thmref{nd-gap} in \secref{sub-gap}), which 
                achieves optimal $O(n^2)$ space and optimal 
                $O(n^3/(B\sqrt{M}))$ cache
                bounds in a cache-oblivious fashion. This
                result improves over prior cache-oblivious
                and cache-efficient algorithm 
                \cite{ChowdhuryRa06, Chowdhury07} on time bound.
        
            \item we give out the first sublinear
                $O(n^{3/4} \log n)$ time and optimal $O(n^3)$ 
                work algorithm for the general GAP problem 
                (\thmref{sub-nd-gap} in \secref{sub-gap}) 
                with non-optimal
                $O(n^3)$ space and $O(n^3/B)$ cache bounds.
                This result improves over Galil and Park's
                sublinear time algorithm 
                \cite{GalilPa94} on work bound, thus answers
                their open question.
        \end{enumerate}
\end{enumerate}
\punt{% begin punt
% Move to the future work section??
However, Galil and Park's original paper did not consider space
and cache requirements of parallel algorithms, which are also
important for a good performance. Though the new GAP algorithm 
designed in this paper has better space and cache bounds, it
is not optimal. It remains an interesting
question how to further bound these two requirements to be 
asymptotically optimal without sacrificing the work and time
bounds.
}% end punt

\punt{% begin punt
\paragrf{Organization: }
\asecref{model} briefly describes the parallel model we use to
calculate all theoretical bounds and recaps the closure method
and ND method; \asecreftwo{sub-1D}{sub-gap} then illustrate
our framework by working on the 1D problem and GAP
problem, respectively; \asecref{concl} concludes the paper 
with an updated open problem, as well as related works.
}% end punt

%% file: contri-table.tex
% \begin{figure*}[htbp]
\begin{figure*}[!ht]
\centering
\begin{minipage}[t]{0.6 \linewidth}
% \scriptsize
\scalebox{0.9}{% begin scalebox
\begin{tabular}{ccccc}
    \toprule
    Algorithm & Work ($T_1$) & Time ($T_\infty$) & Space & Cache ($Q_1$)\\
    \midrule
    Galil and Park's & \multirow{2}{*}{$O(n^2)$} & \multirow{2}{*}{$O(\sqrt{n} \log n)$} & \multirow{2}{*}{$O(n^2)$} & \multirow{2}{*}{$O(n^2/B)$}\\
    final 1D \cite{GalilPa94} & & & & \\
    COP 1D & $O(n^2)$ & $O(n \log n)$ & $O(n)$ & $O(n^2/(BM))$\\
    \emph{space-/cache-efficient} & \multirow{2}{*}{$O(n^2)$} & \multirow{2}{*}{$O(n)$} & \multirow{2}{*}{$O(n)$} & \multirow{2}{*}{$O(n^2/(BM))$}\\
    ND 1D (\thmref{nd-1D}) & & & & \\
    \emph{sublinear-time} & \multirow{2}{*}{$O(n^2)$} & \multirow{2}{*}{$O(\sqrt{n} \log n)$} & \multirow{2}{*}{$O(n^2)$} & \multirow{2}{*}{$O(n^2/B)$}\\
    ND 1D (\thmref{sub-nd-1D}) & & & & \\
    \midrule
    Galil and Park's & \multirow{2}{*}{$O(n^4)$} & \multirow{2}{*}{$O(\sqrt{n} \log n)$} & \multirow{2}{*}{$O(n^4)$} & \multirow{2}{*}{$O(n^4/B)$}\\
    final GAP\cite{GalilPa94} & & & & \\
    COP GAP \cite{ChowdhuryRa06, Chowdhury07} & $O(n^3)$ & $O(n^{\log_2 3})$ & $O(n^2)$ & $O(n^3/(B\sqrt{M}))$\\
    \emph{space-/cache-efficient} & \multirow{2}{*}{$O(n^3)$} & \multirow{2}{*}{$O(n \log n)$} & \multirow{2}{*}{$O(n^2)$} & \multirow{2}{*}{$O(n^3/(B\sqrt{M}))$}\\
    ND GAP (\thmref{nd-gap}) & & & & \\
    \emph{sublinear-time} & \multirow{2}{*}{$O(n^3)$} & \multirow{2}{*}{$O(n^{3/4} \log n)$} & \multirow{2}{*}{$O(n^3)$} & \multirow{2}{*}{$O(n^3/B)$}\\
    ND GAP (\thmref{sub-nd-gap}) & & & & \\
    \bottomrule
\end{tabular}%
}% end scalebox
\caption{Main results of this paper, with comparisons to 
    typical prior works.}
\label{fig:contri-table}
\end{minipage}
\hfil
\begin{minipage}[t]{0.38\linewidth}
\scalebox{0.85}{
% \scriptsize
%\begin{flushright}
% \begin{tabular}{|c|p{4.4cm}|}
\begin{tabular}{ll}
% \hline \hline
\toprule
Notation & Explanation\\
\midrule
% MM  & Matrix Multiplication\\
DP  & Dynamic Programming\\
COP & Cache-Oblivious Parallel\\
% COW & Cache-Oblivious Wavefront\\
% RWS & Randomized Work-Stealing\\
% CAS & Compare-And-Swap\\
ND  & Nested Dataflow\\
$n$ & problem dimension\\
$p$ & \# of cores\\
$\epsilon_i$ & small constant\\
$M$ & cache size\\
$B$ & cache line size\\
$T_1$ & work\\
$T_\infty$ & time (span, depth, critical path length)\\
% $T_p$ & parallel running time on $P$ cores\\
$T_1/T_{\infty}$ & parallelism\\
$Q_1$ & serial cache complexity\\
% $Q_p$ & parallel cache complexity on $P$ threads\\
$\id{a} \parallel \id{b}$ & task \id{b} has \emph{no} dependency on \id{a}\\
$\id{a} \serial \id{b}$ & task \id{b} has \emph{full} dependency on \id{a}\\
$\id{a} \fire \id{b}$ & task \id{b} has \emph{partial} dependency on \id{a}\\
% \hline
\bottomrule
\end{tabular}
}% end scalebox
\caption{Acronyms \& Notation}
\label{fig:symbols}
\end{minipage}
\end{figure*}%

%% file: model.tex
\secput{model}{Theoretical Models}

% 
% This section briefly describes the theoretical models used in 
% our algorithm analysis.
% Briefly describe the theoretic models to calculate the bounds

\paragrf{Parallel Model: }
We adopt the work-time model \cite{JaJa92}
(also known as work-span model \cite{CormenLeRi09})
to calculate work and time complexities.
The model views a parallel computation as a Directed Acyclic 
Graph (DAG). Each vertex stands for a piece of computation 
with no parallel construct and each directed edge 
represents some control or data dependency between the pair 
of vertices.
For simplicity, we count every arithmetic operation such as
multiplication, addition, and comparison uniformly as an 
$O(1)$ operation.
The model calculates an algorithm's time complexity ($T_\infty$) 
by counting the number of arithmetic operations along 
its DAG's critical path. 
\punt{% begin punt
As pointed out in \cite{TangYoKa15, DinhSiTa16}, any extra 
control dependency to data dependency in an algorithm is 
artificial dependency and can only hurt critical path length,
hence should be eliminated by techniques such as Nested Dataflow 
(ND) model \cite{DinhSiTa16}. 
}% end punt
Work ($T_1$) is then the sum over all vertices, with 
parallelism defined as $T_1 / T_\infty$.
Time ($T_\infty$) and work ($T_1$) bounds characterize 
the running time of a parallel algorithm on infinite number 
and one processor(s), respectively.
\punt{% begin punt
This paper assumes a Randomized Work-Stealing (RWS) scheduler 
that has ``busy-leaves'' property as specified in 
\cite{BlumofeJoKu95}. Busy-leaves property 
says that from the time a task is spawned to the time it 
finishes, there is always
at least one subtask from the subcomputation rooted at it 
that is ready. In other words, no leaf task can stall or be 
preempted. Busy-leaves property holds in Cilk
system \cite{IntelCilkPlus10, Leiserson10} and will be 
extended by our \thmref{remote-blocking}.
By RWS scheduler, a better time bound (shorter critical path 
length) means more work available for randomized stealing
at each parallel step at runtime, hence a better
balance.
}% end punt
We call a parallel algorithm \defn{work-efficient} and / or 
\defn{cache-efficient} if its total 
work $T_1$ and / or serial cache bound $Q_1$ matches 
asymptotically that of the best serial 
algorithm for the same problem, respectively.
Analogously, we have the notion of \defn{space-efficient}.
We call a parallel algorithm \defn{sublinear-time} if its
time bound $T_\infty$ is sublinear to the problem dimension 
$n$, i.e. $T_\infty = o(n)$.

\paragrf{Memory Model: }
By the convention of COP algorithms, we calculate only a parallel
algorithm's serial cache complexity, i.e. serialize all its 
parallel 
constructs by some order as if it were run on one 
single computing core, in the ideal cache model 
\cite{FrigoLePr12}.
This simplification is reasonable because 
% \footnote{% begin footnote
a corresponding parallel cache complexity under a 
randomized work-stealing (RWS) runtime scheduler 
can be derived directly by 
$Q_p = Q_1 + O(p T_\infty M/B)$
\cite{AcarBlBl00, SpoonhowerBlGi09}.  
In the rest of paper, the term ``cache bound (complexity)'' 
stands for ``serial cache bound (complexity)'' unless 
otherwise specified.

The ideal cache model has an upper level cache of size $M$ and an 
unbounded lower level memory. Data exchange between the upper 
and lower level is coordinated by an omniscient (offline optimal) 
cache replacement algorithm in cache line of size $B$. It also 
assumes a tall cache, i.e. $M = \Omega(B^{2})$.

To accommodate parallel execution, 
we further assume that the lower level memory follows CREW 
(Concurrent Read Exclusive Write) convention \cite{GibbonsMaRa99}.
Every concurrent 
reads from the same memory location can be accomplished in 
$O(1)$ time, while $n$ concurrent writes to the same memory cell 
have to be serialized by some order
\punt{% begin footnote
The serialization can be accomplished by atomic operations such
as Compare-And-Swap (CAS) or by system's synchronization facility
such as ``\CilkSync'' in the Cilk system.
}% end footnote
and take $O(n)$ total time to accomplish. 
By Brent's theorem \cite{Brent74} 
\punt{% begin punt
(Theorem 1, Sect. 1 of 
\cite{GalilPa94})
}% end punt
, the above work-time model is justified.
\punt{% begin punt
    % remove the punt later?
By \eqref{parallel-cache-bound}, $Q_p - Q_1 = O(P T_\infty M/B)$, 
i.e. the extra cache misses in a parallel execution to its 
serial run is proportional to the critical path length $T_\infty$.
Hence, a shorter critical path length also means less parallel 
cache misses.
}% end punt

% We can possibly add in the brief intro on the basic closure
% method if we have room

%
\punt{% begin punt
\paragrf{The Closure Method: }
Galil and Park \cite{GalilPa94} firstly reduces the 1D and 
GAP problems to the
problem of finding the shortest path of a multi-dimensional
table, i.e. matrix (Lemma 1, Sect. 2.2. of \cite{GalilPa94}), 
then adapts the basic closure method (squaring operation, i.e.
general matrix multiplication on a closed semiring)
to compute shortest path of general graphs \cite{KarpRa90}
to solve the reduced 1D and GAP problems (Lemma 2, Sect. 2.2.
of \cite{GalilPa94}).
\punt{% begin punt
Formally, given a closed semiring $SP = (R, \min, +, +\infty, 0)$, 
where $R$ is a set of non-negative
reals including $+\infty$, $\min$ is the additive operator and  
$+$ is the multiplicative operator with idempotence 
(i.e. $a + a = a$, $\forall a \in R$),
$+\infty$ and $0$ are additive and multiplicative identities, 
respectively. 
If we define a matrix $H$, each of whose entry stands
for the length of edge connecting the two vertices, $+\infty$
if there is no direct connection, 
under the usual definition of matrix operation, the closure 
$H^* = H^n$ (Lemma 2, Sect. 2.2. of \cite{GalilPa94})
then gives out the lengths of shortest paths between all pairs
of entries of the matrix $H$. 
}% end punt
More details will be illustrated
by the 1D example in \secref{sub-1D}.
}% end punt

\paragrf{The Nested Dataflow (ND) Model: }
Tang et al. \cite{TangYoKa15} firstly observe that classic
COP method may introduce excessive control dependency
among recursively derived sub-problems, which un-necessarily
lengthens DAG's critical path. Dinh et al. \cite{DinhSiTa16} 
formalize an ND method based on the observation.
The ND method bears some similarity to pipelining 
technique \cite{JaJa92}, future \cite{Halstead84, Halstead85, BlellochRe97, HerlihyLi14}
and / or synchronization variable \cite{BlellochGiNa97}. 
But instead of explicitly chaining all vertices
in a DAG beforehand, the ND method recursively 
, where the term ``\emph{Nested}'' comes from, refines
and expands the DAG on only data dependency, where the term
``\emph{Dataflow}'' comes from, in a lazy fashion.
Therefore, it not only shortens critical path, i.e. 
time bound, or equivalently
maximizes parallelism, of an algorithm, but also 
achieves cache efficiency in a cache-oblivious fashion 
\cite{FrigoLePr12} because it keeps the recursive 
executing order among vertices.
More discussions on the differences between the ND method
and related works can be found in \cite{DinhSiTa16}.

We describe the ND method in general as follows.
The method employs a new dataflow operator ``$\fire$'' 
(Pronounced ``Fire'') to
address the notion of \defn{partial dependency}. 
If we define a \defn{task} as a \emph{set} of vertices of a DAG,
a partial dependency between a pair of tasks \id{a} and \id{b} 
, denoted by``$\id{a} \fire \id{b}$'',
indicates that a subset of subtasks of \id{b} depends
on a subset of subtasks of \id{a}. 
That is to say, \id{b} will get notified by some signal 
, from either runtime system or 
\id{a} depending on implementation, and can start executing 
some subset of its subtasks when their data dependency 
are satisified. 
If an algorithm follows a divide-and-conquer framework,
as many COP algorithms does, the ``$\fire$'' operator will
preserve 
recursive executing order among subtasks of \id{a} and \id{b}
 to attain cache efficiency in a cache-oblivious fashion.
A partial dependency will be refined, recursively if it's a 
recursive algorithm, by fire rules throughout computation. 
\punt{% begin punt
\footnote{% begin footnote
In a practical manner, a computation at runtime does not 
have to refine and expand all vertices of its DAG beforehand,
it can refine partial dependencies and expand the DAG
in a lazy fashion. That is, it does not have to expand
a task until at least one of its subtasks can be expanded, i.e.
some subtask having all its input data available, and does not
have to refine a partial dependency $\id{a} \fire \id{b}$ until
either \id{a} or \id{b} can be expanded.
}% end footnote
}% end punt
%
% according to only data dependency. 
%
A partial dependency will stop refinement until both its 
source and sink tasks become leaf vertices (vertices with no 
parallel construct), in which case the ``$\fire$'' operator 
reduces to a classic ``$\serial$'' (Serial) construct. 
By the definition, the classic ``$\parallel$'' (Parallel) 
and ``$\serial$'' (Serial)
constructs are just syntactic sugar for the two extreme cases. 
That is, notation ``$\id{a} \parallel \id{b}$'' indicates 
that no subtasks of \id{b} depends on any subtasks of \id{a}, 
i.e. \emph{no} dependency, while ``$\id{a} \serial \id{b}$'' 
says that all subtasks of \id{b} depend on all subtasks of 
\id{a}, i.e. a \emph{full} dependency. 

%% file: sublinear-lws.tex
\secput{sub-ND-Algo}{Work-Efficient and Sublinear-Time 
Nested Dataflow Algorithms}

\subsecput{sub-1D}{The 1D Problem}

% Organization of this section
\paragrf{Organization: }
We firstly recap Galil and Park's sublinear-time 
algorithm \cite{GalilPa94} by \thmref{old-sublinear-1D},
which serves as the foundation for later discussions; 
Then we address several components, i.e. 
a classic COP algorithm by \lemref{cop-1D},
an improved ND version by \thmref{nd-1D}.
Finally, we present our main \thmref{sub-nd-1D} for
the 1D problem constructed from these components.

\subsubsecput{old-sublinear-1D}{Old Work-Efficient and 
Sublinear-Time 1D algorithm}
\input{1D-algo-pic}

The following algorithm is extracted from Sect. 2 of Galil and
Park's original work \cite{GalilPa94}. 
They firstly reduce the problem of solving the recurrences of 
\eqref{1D} to a shortest path problem, 
then solve the shortest path by squaring (the closure method,
i.e. general matrix multiplication on a semiring). 
Finally, they 
reduce the amount of work by the method of indirection.

The main steps are recapped as follows.
It defines a matrix $H$ as follows such that each entry 
$H(i, j)$ denotes the current 
shortest path from coordinate $i$ to $j$.
\begin{align}
H(i, i) &= 0 & & \text{for $0 \leq i \leq n$} \nonumber \\
H(i, j) &= w(i, j) & & \text{for $0 \leq i < j \leq n$} \label{eq:1D-H}\\
H(i, j) &= +\infty & & \text{for $i > j$} \nonumber
\end{align}
It then defines a squaring operation on $H$ as follows.
\begin{align}
    H^2 (i, j) &= \min_{\substack{i \leq r \leq j}} \{H(i, r) + H(r, j)\} 
    & & \text{for $0 \leq i < j \leq n$} \label{eq:1D-squaring} 
\end{align}
$H^k (i, j)$ is then the length of shortest path from $i$ to $j$
via at most $k$ edges, and the closure $H^*$ ($H^* = H^n$, 
according
to Lemma 2 of \cite{GalilPa94}) contains the lengths of shortest 
paths between all pairs of coordinates. Hence, solving the 1D
recurrence 
reduces to finding the shortest path from $0$ to $n$, i.e. from
$(0, 0)$ to $(0, n)$, which is depicted by the top shaded row in 
\figref{sublinear-1D}.

The squaring operation is basically a general matrix 
multiplication (MM) on a closed semiring 
and takes the computational overheads in \lemref{mm-n3-space}
\cite{CormenLeRi09} to accomplish.  
\begin{lemma} [\cite{CormenLeRi09}]
    General MM of size $n$, i.e. an $n$-by-$n$
    matrix multiplies another $n$-by-$n$ matrix, on a closed
    semiring can be computed
    in $O(\log n)$ time, $O(n^3)$ work, $O(n^3)$ space
    and $O(n^3/B)$ cache misses.
    \label{lem:mm-n3-space}
\end{lemma}
%
% \punt{% begin punt
\begin{proof}
The lemma can be proved by constructing a $2$-way
divide-and-conquer algorithm that recursively 
divides a dimension-$n$ ($n$-by-$n$-by-$n$)
matrix multiplication (MM) into eight ($8$) concurrent 
dimension-$(n/2)$ sub-MMs by allocating temporary matrix
of size $n$-by-$n$ to hold intermediate results and 
merge the results by addition in the end of each recursion.
\punt{% begin punt
The divide-and-conquer pattern is shown in \eqref{mm-dac}.
\begin{align}
    \begin{bmatrix}
        C_{00} & C_{01} \\
        C_{10} & C_{11}
    \end{bmatrix}
    &= 
    \begin{bmatrix}
        A_{00} & A_{01} \\
        A_{10} & A_{11}
    \end{bmatrix} 
    \otimes
    \begin{bmatrix}
        B_{00} & B_{01} \\
        B_{10} & B_{11}
    \end{bmatrix} \\
    &=
    \begin{bmatrix}
        A_{00} \otimes B_{00} & A_{00} \otimes B_{01} \\
        A_{10} \otimes B_{00} & A_{10} \otimes B_{01}
    \end{bmatrix} \\
    & \qquad \oplus 
    \begin{bmatrix}
        A_{01} \otimes B_{10} & A_{01} \otimes B_{11} \\
        A_{11} \otimes B_{10} & A_{11} \otimes B_{11}
    \end{bmatrix}
\label{eq:mm-dac}
\end{align}
}% end punt
\end{proof}
% }% end punt

\begin{lemma} [Lemma 3 of \cite{GalilPa94}]
Referring to \figref{1D-square-decomp}, a square $Q^*$ of 
size $v$, i.e. of size $v$-by-$v$, can be computed from 
two adjacent triangles $T_1$, $T_2$, and $Q$ by 
$Q^* = T_1 Q T_2$ in $O(\log v)$ time, $O(v^3)$ space 
and work, and $O(v^3/B)$ cache bounds.
\label{lem:1D-square-decomp}
\end{lemma}
\begin{proof}
Assuming that the indices for $Q^*$, $T_1$, and $T_2$ are 
from $1$ to $m$, i.e.
\begin{align}
    T_1 &= H^* (i, j) & & \text{for $1 \leq i < j \leq m$} \nonumber \\
    T_2 &= H^* (i, j) & & \text{for $m/2 < i < j \leq m$} \nonumber\\
    Q & = H (i, j)    & & \text{for $i \leq m/2 < j$} \nonumber\\
    Q^* &= H^* (i, j) = T_1 Q T_2 & & \text{for $i \leq m/2 < j$} \label{eq:1D-H-star}
\end{align}
The semantics of $Q^*$'s computation stands for computing the 
shortest path from $i$ to $j$ by taking the minimum of all 
possible paths $[(i, p), (p, q), (q, j)]$, where $(i, p)$ is a 
cell in $T_1$ representing the shortest path from 
$i$ to $p$, $(p, q)$ is a cell in $Q$ standing for an edge 
from $p$ to $q$, and $(q, j)$ is a cell in $T_2$ representing
the shortest path from $q$ to $j$.
Since the computation of $T_1 Q T_2$ requires two squaring 
operations, each of which costs $O(\log v)$ time and $O(v^3)$ 
work and space, and $O(v^3/B)$ cache according to 
\lemref{mm-n3-space}, the lemma then follows.
\end{proof}

\begin{theorem}
[Refinement of Theorem 2 of \cite{GalilPa94}] 
There is an algorithm that solves the 1D recurrences of 
\eqref{1D} in a sublinear 
$O(\sqrt{n} \log n)$ time, optimal $O(n^2)$ work, $O(n^2)$ space, 
and $O(n^2/B)$ cache bounds.
\label{thm:old-sublinear-1D}
\end{theorem}
\begin{proof}
Referring to \figref{sublinear-1D}, the algorithm can be
formulated as follows:
\begin{enumerate}
    % \punt{% begin punt
    \item Dividing the top row of $H^*$ into $n/v$ intervals, 
        where $v$ is a parameter to determine later. 
    % }% end punt

    \item Computing $n/v$ shaded triangles of $H^*$ on 
        diagonal bounded by cells $(i, i)$, $(i, i + v)$, 
        and $(i + v, i + v)$ for $0 \leq i \leq n/v - 1$, 
        simultaneously by repeated squaring as the algorithm in
        \lemref{1D-square-decomp}.
        Observing that each top-level triangle
        is computed recursively in $\log v$ levels
        and that top-level computation dominates, 
        the overheads to compute one top-level triangle then
        sum up to
        $O(\log^2 v)$ time, $O(v^3)$ work and space, and 
        $O(v^3/B)$ cache misses. 
        Summing up the overheads for $n/v$ concurrent top-level
        triangles yield the costs of $O(\log^2 v)$ time,
        $O(nv^2)$ work and space, and $O(nv^2/B)$ cache misses.

    \item Computing the top shaded row of $H^*$ from the left-most 
        interval to right-most in $n/v$ iterations.
        For convenience, let's denote the top shaded row of 
        $H^*$ by $f$.
        The first interval, i.e. $f(0), \ldots, f(v)$, is given 
        by the top row of first shaded triangle for free; 
        For $l \in [2, n/v)$, we compute the $l$-th interval
        , i.e.  $f((l-1)v+1), \ldots, f(lv)$, by the 
        squaring of following three matrices, 
        \begin{enumerate}
            \item all previous intervals, i.e. $f(0), \ldots, 
                f((l-1)v)$, which is a $1$-by-$(l-1)v$ matrix.
                We call it ``$H^*_{-}$'' for convenience.
            \item the $l$-th interval of $H$, i.e. the rectangular 
                region of $H$ bounded by cells 
                % FIXME: We may want to mark the four corners
                % in the \figref{sublinear-1D}
                $(0, (l-1)v + 1)$, $((l-1)v + 1, (l-1)v + 1)$,
                $((l-1)v + 1, lv)$, and $(0, lv)$, which is an 
                $(l-1)v$-by-$v$ matrix. 
                We call it ``$H_{\Box}$''.
                Note that matrix $H$ can be computed on-the-fly in 
                $O(1)$ time with no memory access.
            \item the shaded triangle on the same column, i.e.
                the $H^*$ region bounded by cells 
                $((l-1)v, (l-1)v)$, $((l-1)v, lv)$, and 
                $(lv, lv)$, which is a $v$-by-$v$ triangular 
                matrix. We call it ``$H^*_{\triangle}$''.
        \end{enumerate}
        The first squaring of $H^*_{-}$ with $H_{\Box}$ yields 
        an intermediate $1$-by-$v$ matrix $H'_{-}$ 
        in $O(\log (lv))$ time, $O(lv^2)$ space and work, and 
        $O(lv^2/B)$ cache misses; The second squaring of the
        intermediate $H'_{-}$ with $H^*_{\triangle}$ 
        yields the final $1$-by-$v$ $l$-th interval
        in $O(\log v)$ time, $O(v^2)$ space and work, 
        and $O(v^2/B)$ cache misses.
        \punt{% begin punt
        Summing up the overheads of the two squarings yields 
        a cost of $O(\log (lv))$ time, $O(lv^2)$ space and 
        work, and $O(lv^2/B)$ cache misses. 
        }% end punt
        Apparently, the first squaring dominates.

        Summing up over $l \in [2, n/v]$ iterations yields a 
        cost of
        $O((n/v) \log n)$ time, $O(nv)$ space, with temporary 
        space for squaring of \lemref{mm-n3-space} reused 
        across iterations, $O(n^2)$ work, and 
        $O(n^2/B)$ cache misses.
\end{enumerate}

Summing up the overheads of the above two steps and making 
$v = \sqrt{n}$ yields the conclusion.
\end{proof}

Though this algorithm is work-efficient, it's neither 
space- nor cache-efficient since a straightforward 
cache-oblivious algorithm requires only $O(n)$ space and incurs 
$O(n/B + n^2/(BM))$ cache misses.

\subsubsecput{co-1D}{Cache-Oblivious Parallel 1D Algorithm}
Referring to \figref{co-1D}, a straightforward cache-oblivious
parallel (COP)
algorithm recursively divides the work into three 
or four quadrants depending on the shape and schedules their
executing order according to the data dependencies in the 
granularity of \emph{quadrants}.
We denote the top-left quadrant by $(0, 0)$, top-right 
$(0, 1)$, bottom-left $(1, 0)$, and bottom-right $(1, 1)$. 
Referring to the pseudo-code in \figref{cop-1D-tri},
the serialization between the computation of $A_{00}$ and 
$A_{01}$ is because the computation of $A_{01}$ requires
the results of $A_{00}$ as input. If an algorithm does not 
allocate temporary space, the computation of $A_{11}$ has to lay 
behind $A_{01}$ by our CREW assumption because they output to 
the same region.
A similar analysis applies to the pseudo-code in 
\figref{co-1D-box}. 

\begin{lemma}
There is a COP algorithm that solves
the 1D recurrences of \eqref{1D} in $O(n \log n)$ time, 
optimal $O(n^2)$ work, 
optimal $O(n)$ space, and optimal $O(n/B + n^2/(BM))$ cache bounds.
\label{lem:cop-1D}
\end{lemma}
\begin{proof}
The COP algorithm is given in 
\figref{co-1D} and we have following recurrences for its 
time and cache complexity, respectively. 
The subscripts of $\tri$ and $\Box$ in the recurrences
stand for the $\proc{cop-1D}_\tri$ and $\proc{co-1D}_\Box$
algorithms in \figreftwo{cop-1D-tri}{co-1D-box} respectively.
\aeqref{cop-1D-cache-stop} is the stop condition of cache bound. 
The recurrences solve to $T_{\infty, \tri} (n) = O(n \log n)$ and 
$Q_{1, \tri} (n) = O(n/B + n^2/(BM))$, with an optimal $O(n)$ 
space bound because the algorithm does not use temporary space.
\begin{align}
    T_{\infty, \tri}(n) &= 2 T_{\infty, \tri} (n/2) + T_{\infty, \Box} (n/2) \\
    T_{\infty, \Box} (n) &= 2 T_{\infty, \Box} (n/2) \label{eq:cop-1D-box} \\
    Q_{1, \tri} (n) &= 2 Q_{1, \tri} (n/2) + Q_{1, \Box} (n/2) \\
    Q_{1, \Box} (n) &= 4 Q_{1, \Box} (n/2) \\
    Q_{1, \tri} (n) &= Q_{1, \Box} (n) = O(n/B) \quad \text{if $n \leq \epsilon_6 M$} \label{eq:cop-1D-cache-stop}
\end{align}
\end{proof}

\input{co-1d-code}

\subsubsecput{nd-1D}{Nested Dataflow 1D algorithm}
We can improve the time bound of above COP algorithm
by refining the data dependency of $\proc{cop-1D}_{\tri}$ 
algorithm recursively by the ND method as shown in 
\figref{nd-1D-tri}. 
The ``$\tfire{\tri\Box}$'' construct in figure indicates a 
partial dependency between the source and sink tasks, which
is specified by the fire rule in \figref{nd-1D-fire-rule}. 
The $\oplus$ and $\ominus$ notation are wildcards
to match at runtime source and sink tasks respectively. 
The fire rule says that the $(0, 0)$, $(0, 1)$ subtasks
nested in the same sink task only partially depends on the 
$(0, 0)$ subtask of the source, respectively the $(1, 0)$ 
and $(1, 1)$ subtasks of the same sink partially depends on the 
$(1, 1)$ subtask of the source. The partial dependences will 
then be 
recursively refined by the same rule until base cases where the 
``$\tfire{\tri\Box}$'' construct will reduce to a ``$\serial$'' 
construct.

\begin{theorem}
    There is an ND algorithm that solves the 1D recurrences of 
    \eqref{1D}
    in $O(n)$ time, optimal $O(n^2)$ work, $O(n)$ space and 
    $O(n^2/(BM))$ cache bounds.
\label{thm:nd-1D}
\end{theorem}
\begin{proof}
Referring to \figref{nd-1D-tri}, the ND algorithm just 
re-schedules subtasks to run as soon as their input data are 
ready so that it does not use more space or incur more
cache misses than the classic COP counterpart in 
\figref{cop-1D-tri}.
Its new time recurrences are as follows, which solve 
to $O(n)$.
\begin{align}
T_{\infty, \tri}(n) &= T_{\infty, \tfire{\tri\Box}} (n/2) + T_{\infty, \tri} (n/2) \\
T_{\infty, \tfire{\tri\Box}} (n) &= 2 T_{\infty, \tfire{\tri\Box}} (n/2) \label{eq:nd-1D-tri}\\
T_{\infty, \tfire{\tri\Box}} (1) &= T_{\infty, \tri} (1) + T_{\infty, \Box} (1) \label{eq:nd-1D-fire-stop} 
\end{align}
\end{proof}

\begin{theorem}
    There is an ND algorithm that solves the 1D recurrences of 
    \eqref{1D}
    in $O(\sqrt{n} \log n)$ time, optimal $O(n^2)$ work, $O(n^2)$ 
    space and $O(n^2/B)$ cache bounds.
\label{thm:sub-nd-1D}
\end{theorem}
\begin{proof}
We construct the algorithm based on following observations.
\begin{enumerate}
    \item The $\tfire{\tri\Box}$ construct in \eqref{nd-1D-tri} 
        will invoke some kernel functions to compute triangular 
        and rectangular shape regions respectively when the
        recursion goes down to base cases. 
        
    \item Galil and Park's algorithm in \thmref{old-sublinear-1D} 
        computes a triangular shape region in sublinear time
        and optimal work. 
        
    \item By using extra temporary space, i.e. like the general 
        MM algorithm in \lemref{mm-n3-space}, we can have a 
        sublinear-time and optimal-work algorithm to compute 
        rectangular shape regions as well. 

    \item Combining the two sublinear-time and optimal-work
        kernel functions for triangular and
        rectangular shape regions respectively, we can stop the 
        recursion of \eqref{nd-1D-tri} earlier.
\end{enumerate}

We allocate temporary space for the $\proc{co-1D}_{\Box}$ 
algorithm in \figref{co-1D-box} and have following 
$\proc{sub-1d}_{\Box}$ algorithm as shown in \figref{sub-1D-box}.
\input{sub-1d-code}

In \figref{sub-1D-box}, \lirefs{sub-1D-parallel}{sub-1D-sync} 
parallelizes the task of
updating $A$ from $B$ to four subtasks by allocating 
temporary space of $A'$, which is of the same size of $A$. 
The ``$\serial$'' construct on 
\liref{sub-1D-sync} is a synchronization operation so that
the overall task can not proceed until all four concurrent 
subtasks are done. 
Note that the merge by addition on \liref{sub-1D-add} can
be parallelized and accomplished in $O(1)$ time if counting
only data dependency. The recurrence
for $\proc{sub-1D}_{\Box}$'s time bound is then revised to 
\eqref{sub-1D-box}, which solves to $O(\log n)$.
\begin{align}
    T_{\infty, \Box}(n) &= T_{\infty, \Box} (n/2) + O(1) \label{eq:sub-1D-box}
\end{align}

Supplying this $\proc{sub-1D}_{\Box}$ for rectangular shape
regions and Galil and Park's final algorithm for triangular shape
regions to stop the recursion of \eqref{nd-1D-tri} earlier,
we have following stop condition of 
\eqref{sub-nd-1D-fire-stop} to replace
\eqref{nd-1D-fire-stop}, where $v$ is a parameter to be determined
later.
\begin{align}
    T_{\infty, \tfire{\tri\Box}} (v) &= T_{\infty, \tri} (v) + T_{\infty, \Box} (v) \label{eq:sub-nd-1D-fire-stop}\\
    T_{\infty, \tri}(v) &= O(\sqrt{v} \log v) \\
    T_{\infty, \Box}(v) &= O(\log v)
\end{align}
$T_{\infty, \tfire{\tri\Box}}(v)$ then solves to 
$O(\sqrt{v} \log v)$.
Supplying this result into \eqref{nd-1D-tri}, we have
$T_{\infty, \tri} (n) = O(n/v \cdot \sqrt{v} \log v) + 
T_{\infty, \tri} (n/2) = O(n/\sqrt{v} \log v)$. By making
$v = n$, the time bound solves to $O(\sqrt{n} \log n)$.
It's easy to verify that $\proc{sub-1D}_{\Box}$ has the same
space and cache bounds as Galil and Park's final algorithm 
for triangular regions (\thmref{old-sublinear-1D}). The
overall space and cache bounds then are a simple summation
over all regions.
\end{proof}

\paragrf{Discussions: }
Our result of \thmref{nd-1D} improves over prior
cache-oblivious and cache-efficient algorithm of \lemref{cop-1D}
on time bound without sacrificing work, space, and cache 
efficiency.
Our result of \thmref{sub-nd-1D} gives out another dimension
of tradeoff by reducing further time bound to be sublinear
but at the cost of increasing the space and cache bounds a bit
, actually still be asymptotically the same as that of Galil and 
Park's final algorithm (\thmref{old-sublinear-1D}).
Moreover, \thmref{sub-nd-1D} provides new insights into
solving DP recurrences with more than $O(1)$ dependency 
and will show its power in solving the GAP problem in 
\secref{sub-gap}. 

%% file: 1D-algo-pic.tex
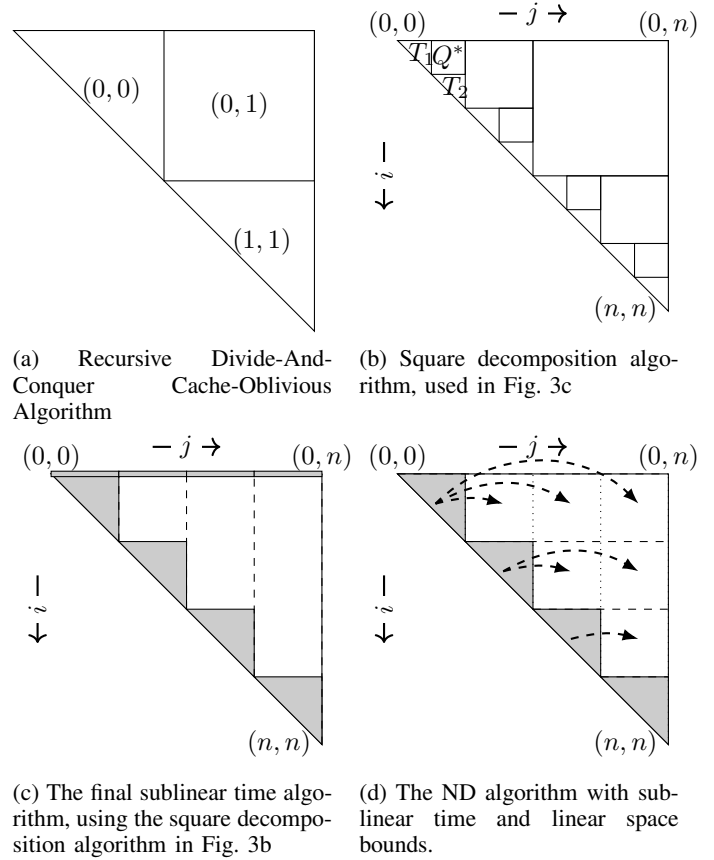
\begin{figure}[!ht]
\begin{subfigure}[t]{0.48 \linewidth}
\begin{tikzpicture}[scale = 2]
\path[draw] (0, 2) -- (2, 2) -- (2, 0) -- cycle;
\draw (1, 1) rectangle (2, 2);
\node [xshift = 3mm, yshift = 2mm] at (0.5, 1.5) {$(0, 0)$};
\node at (1.5, 1.5) {$(0, 1)$};
\node [xshift = 3mm, yshift = 2mm] at (1.5, 0.5) {$(1, 1)$};
\end{tikzpicture}
\caption{Recursive Divide-And-Conquer Cache-Oblivious Algorithm}
\label{fig:co-1D}
\end{subfigure}
\hfill
\begin{subfigure}[t]{0.48 \linewidth}
\begin{tikzpicture}[scale = 1.8]
\draw [thick, ->] (0.75, 2.2) -- (1.25, 2.2) node[fill=white, pos=.5]{$j$}; 
\draw [thick, ->] (-0.1, 1.25) -- (-0.1, 0.75) node[fill=white, rotate=90, pos=.5]{$i$}; 
\node[above, yshift=-0.8mm] at (0, 2) {$(0, 0)$};
\node[above, yshift=-0.8mm] at (2, 2) {$(0, n)$};
\node[left, xshift=-0.2mm, yshift=0.2mm ] at (2, 0) {$(n, n)$};
\path[draw] (0, 2) -- (2, 2) -- (2, 0) -- cycle;
\draw (1, 1) rectangle (2, 2);
\foreach \x in {0.5, 1.5} {
    \draw (\x, 2 - \x) rectangle (\x + 0.5, 2 - \x + 0.5);
}
\foreach \x in {0.25, 0.75, 1.25, 1.75} {
    \draw (\x, 2 - \x) rectangle (\x + 0.25, 2 - \x + 0.25);
}
\node [xshift = 1mm, yshift = 0.5mm] at (0.125, 2 - 0.125) {$T_1$};
\node [xshift = 1mm, yshift = 0.5mm] at (0.125 + 0.25, 2 - 0.125 - 0.25) {$T_2$};
\node at (0.125 + 0.25, 2 - 0.125) {$Q^*$};
\end{tikzpicture}
\caption{Square decomposition algorithm, used in \figref{sublinear-1D}}
\label{fig:1D-square-decomp}
\end{subfigure}
\vfill
\begin{subfigure}[t]{0.48 \linewidth}
\begin{tikzpicture}[scale = 1.8]
\draw [thick, ->] (0.75, 2.2) -- (1.25, 2.2) node[fill=white, pos=.5]{$j$}; 
\draw [thick, ->] (-0.1, 1.25) -- (-0.1, 0.75) node[fill=white, rotate=90, pos=.5]{$i$}; 
\node[above, yshift=-0.8mm] at (0, 2) {$(0, 0)$};
\node[above, yshift=-0.8mm] at (2, 2) {$(0, n)$};
\node[left, xshift=-0.2mm, yshift=0.2mm ] at (2, 0) {$(n, n)$};
\path[draw] (0, 2) -- (2, 2) -- (2, 0) -- cycle;
\foreach \x in {0, 0.5, 1, 1.5} {
    \filldraw[fill=gray!40] (\x, 2 - \x) -- (\x + 0.5, 2 - \x) -- (\x + 0.5, 1.5 - \x) -- cycle;
    \draw [dashed] (\x + 0.5, 1.5 - \x) -- (\x + 0.5, 2);
    \filldraw[fill=gray!40] (\x, 2 - 0.02) rectangle (\x + 0.5, 2 + 0.02);
}
\end{tikzpicture}
\caption{The final sublinear time algorithm, using the square 
decomposition algorithm in \figref{1D-square-decomp}}
\label{fig:sublinear-1D}
\end{subfigure}
\hfill
\begin{subfigure}[t]{0.48 \linewidth}
\begin{tikzpicture}[scale = 1.8]
\draw [thick, ->] (0.75, 2.2) -- (1.25, 2.2) node[fill=white, pos=.5]{$j$}; 
\draw [thick, ->] (-0.1, 1.25) -- (-0.1, 0.75) node[fill=white, rotate=90, pos=.5]{$i$}; 
\node[above, yshift=-0.8mm] at (0, 2) {$(0, 0)$};
\node[above, yshift=-0.8mm] at (2, 2) {$(0, n)$};
\node[left, xshift=-0.2mm, yshift=0.2mm ] at (2, 0) {$(n, n)$};
\path[draw] (0, 2) -- (2, 2) -- (2, 0) -- cycle;
\foreach \x [count=\xi] in {0, 0.5, 1, 1.5} {
    \filldraw[fill=gray!40] (\x, 2 - \x) -- (\x + 0.5, 2 - \x) -- (\x + 0.5, 1.5 - \x) -- cycle;
    \draw [dashed] (\x, 2 - \x) -- (2, 2 - \x);
    \draw [dotted] (\x + 0.5, 1.5 - \x) -- (\x + 0.5, 2);
}    
\foreach \x [count=\xi] in {0, 0.5, 1} {
    \foreach \y [count=\yi] in {\xi, ..., 3} {
        \draw [thick, dashed, path arrow] (\x + 0.28, 2 - \x - 0.22) to [out = \yi * 15, in = 180 - \yi * 15] (\x + 0.28 + \yi * 0.5, 2 - \x - 0.22);
    }
}    
\end{tikzpicture}
\caption{The ND algorithm with sublinear time and linear space 
bounds.}
\label{fig:nd-1D}
\end{subfigure}
\caption{Different Algorithms for the 1D problem}
\label{fig:1D}
\end{figure}

%% file: co-1d-code.tex
\begin{figure}[!ht]
\begin{subfigure}[b]{0.48 \linewidth}
%\scriptsize
\begin{codebox}
\Procname{$\proc{cop-1D}_{\tri}(A)$}
\li $\proc{cop-1D}_{\tri}(A_{00}) \serial$
\li $\proc{co-1D}_{\Box}(A_{01}, A_{00}) \serial$
\li $\proc{cop-1D}_{\tri}(A_{11}) \serial$
\li \Return
\end{codebox}
\vspace*{-1 \baselineskip}
\caption{The COP algorithm for a triangular region}
\label{fig:cop-1D-tri}
\end{subfigure}
\hfil
\begin{subfigure}[b]{0.48 \linewidth}
%\scriptsize
\begin{align*}
    \oone{+} \tfire{\tri\Box} &\oone{-} = \{ \\
         & \oone{+}_{00} \tfire{\tri\Box} \{\oone{-}_{00}, \oone{-}_{01}\} \\
         & , \oone{+}_{11} \tfire{\tri\Box} \{\oone{-}_{10}, \oone{-}_{11}\}\}
\end{align*}
\vspace*{-1 \baselineskip}
\caption{The fire rule of ND algorithm}
\label{fig:nd-1D-fire-rule}
\end{subfigure}
\vfil
\begin{subfigure}[b]{0.48 \linewidth}
%\scriptsize
\begin{codebox}
\Procname{$\proc{co-1D}_{\Box}(A, B)$}
\li $\proc{co-1D}_{\Box}(A_{l}, B_{l})$
\zi \quad $\parallel \proc{co-1D}_{\Box}(A_{r}, B_{r}) \serial$
\li $\proc{co-1D}_{\Box}(A_{l}, B_{r})$
\zi \quad $\parallel \proc{co-1D}_{\Box}(A_{r}, B_{l}) \serial$
\li \Return
\end{codebox}
\vspace*{-1 \baselineskip}
\caption{The COP algorithm for a rectangular region}
\label{fig:co-1D-box}
\end{subfigure}
\hfil
\begin{subfigure}[b]{0.48 \linewidth}
%\scriptsize
\begin{codebox}
\Procname{$\proc{nd-1D}_{\tri}(A)$}
\li $\proc{nd-1D}_{\tri}(A_{00}) \tfire{\tri\Box}$
\zi \quad $\proc{co-1D}_{\Box}(A_{01}, A_{00}) \serial$
\li $\proc{nd-1D}_{\tri}(A_{11}) \serial$
\li \Return
\end{codebox}
\vspace*{-1 \baselineskip}
\caption{The ND algorithm for a triangular region}
\label{fig:nd-1D-tri}
\end{subfigure}
\caption{Pseudo-codes of cache-oblivious 1D algorithms}
\label{fig:co-1D-code}
\end{figure}

%% file: sub-1d-code.tex
\begin{figure}[!ht]
%\scriptsize
\begin{codebox}
\Procname{$\proc{sub-1D}_{\Box}(A, B)$}
\li \Comment Update region $A$ from $B$
\li \Comment Allocate temporary space
\li $A' \assign \func{alloc}(\func{sizeof}(A))$ \label{li:sub-1D-alloc}
\li $\proc{sub-1D}_{\Box}(A_{l}, B_{l}) \parallel \proc{sub-1D}_{\Box}(A_{r}, B_{r})$ \label{li:sub-1D-parallel}
\li $\parallel \proc{sub-1D}_{\Box}(A'_{l}, B_{r}) \parallel \proc{sub-1D}_{\Box}(A'_{r}, B_{l}) \serial$ \label{li:sub-1D-sync}
\li \Comment Merge $A'$ to $A$ by addition
\li $A = A + A'$ \label{li:sub-1D-add}
\li \func{free}(\id{A}')
\li \Return
\end{codebox}
\vspace*{-1 \baselineskip}
\caption{A Sublinear-Time and Optimal-Work algorithm for the
rectangular shape region of 1D problem}
\label{fig:sub-1D-box}
\end{figure}

%% file: sublinear-gap.tex
\subsecput{sub-gap}{The GAP Problem}
%
% Organization of this section
\paragrf{Organization: }
We firstly recap Galil and Park's final algorithm 
by \thmref{old-sublinear-gap}, which serves as the foundation 
for later discussion; 
Then we address the classic COP 
algorithm developed by Chowdhury and Ramachandran
\cite{ChowdhuryRa06, Chowdhury07}, as well as its ND 
improvement by \thmref{nd-gap};
Finally, we combine all components to yield the main
\thmref{sub-nd-gap}. 

\subsecput{old-sublinear-gap}{Sublinear-Time but Non-Work-Efficient 
GAP algorithm}
The key of Galil and Park's sublinear-time algorithm (Sect. 
4 of \cite{GalilPa94}) is similar to that of the 1D 
algorithm (\secref{old-sublinear-1D} of this paper).
That is, firstly, they reduce the GAP problem to a 
shortest path problem, then solve the shortest path
problem by the closure method. Finally, they reduce 
the amount of total work by indirection.

\input{gap-algo-pic}

The main steps as follows are essentially a $2$D
version of the closure method for the 1D problem in 
\secref{old-sublinear-1D}.
Firstly, it defines matrix $H$ as follows such that each
entry $H(p, q, i, j)$, where $p < i$ and $q < j$, is 
the current shortest path from cell $(p, q)$ to $(i, j)$ via 
one intermediate cell $(p, j)$ or $(i, q)$ as shown in 
\figref{gap-square-decomp}. 
$H(\cdot, \cdot, \cdot, \cdot)$ can
be viewed as a combination of two independent matrices
like the one defined by \eqref{1D-H} for the 1D problem, 
and is a $4$D array by itself.
\begin{align}
    H(i, j, i, j) &= 0 & \text{$0 \leq i, j \leq n$} \nonumber\\
    H(p, q, i, j) &= w(q, j) + w'(p, i) & \text{$p < i-1 \land q < j-1$} \label{eq:gap-H}\\
    H(p, q, i, j) &= +\infty & \text{for all others} \nonumber 
\end{align}
\begin{align*}
    & H(i-1, j-1, i, j) = \min\{s_{ij}, w(j-1, j) + w'(i-1, i)\} 
\end{align*}
It then defines a squaring operation on $H$ as follows.
\begin{align}
    % H^2(p, q, i, j) &= \min_{\substack{p \leq s \leq i\\q \leq r \leq j}}\{ H(p, q, s, r) + H(s, r, i, j) \} \label{eq:squaring-gap}&
H^2(p, q, i, j) &= \min_{\begin{subarray}{c}p \leq s \leq i\\q \leq r \leq j\end{subarray}}\{ H(p, q, s, r) + H(s, r, i, j) \} \label{eq:gap-squaring}
\end{align}
\begin{lemma}
    The squaring operation of \eqref{gap-squaring} on two $4$D 
    matrices of dimension $n$ as defined by \eqref{gap-H} costs 
    $O(\log n)$ time, 
    $O(n^6)$ work and space with $O(n^6/B)$ cache misses.
\label{lem:gap-squaring}
\end{lemma}
\begin{proof}
Referring to \eqref{gap-squaring}, since the squaring operation 
chooses
the two coordinates of intermediate cell $(s, r)$ independently
from $p$ to $i$ and $q$ to $j$, 
the update to any cell requires a $\min$ operation on
$O(n^2)$ intermediate results of $+$ operations. 
By allocating temporary space like that of  
\lemref{mm-n3-space}, it takes $O(\log n)$ time, 
$O(n^2)$ work and space with $O(n^2/B)$ cache misses
to update one cell. 
Summing up over the $O(n^4)$ cells of $H$ yields the 
conclusion.
\end{proof}
By Lemma $2$ in \cite{GalilPa94}, $H^* = I + H + H^2 + \cdots + 
H^n = H^n$ is then the solution to the GAP recurrence of 
\eqref{gap} and can 
be computed by $\log n$ repeated squaring on $H$. 
By \lemref{gap-squaring}, this straightforward computation takes 
totally
$O(\log^2 n)$ time, $O(n^6 \log n)$ work, $O(n^6)$ space,
and $O(n^6 \log n / B)$ cache misses by reusing temporary space 
across repeated squarings. 

Following \lemref{gap-square-decomp} reduces the total 
work by a $2$D reduction technique (The $1$D version
is in the proof of \lemref{1D-square-decomp}).

\begin{lemma} 
There is a $2$D square decomposition algorithm that computes 
the GAP recurrences of \eqref{gap} in $O(\log^2 n)$ time, 
$O(n^6)$ work and space, and $O(n^6/B)$ cache misses.
\label{lem:gap-square-decomp}
\end{lemma}
\begin{proof}
Referring to \figref{gap-square-decomp}, the algorithm 
computes the GAP recurrence by recursively decomposing the
$2$D region of $H$ (a $2$D projection of $H$) into four quadrants 
and computes $H^*$ bottom up 
from the smallest squares to the largest in $\log n$ levels. 
At any level $k$, $H^*$ of $(n/2^k)^2$ squares of size $2^k$
are computed as follows.
To compute a square $Q$ from four quadrants $Q_{00}$, $Q_{01}$, 
$Q_{10}$, and $Q_{11}$, the
algorithm firstly computes $H^*_{00, 01} = H^*_{00} H_{00, 01} 
H^*_{01}$ and $H^*_{10, 11} = H^*_{10} H_{10, 11} H^*_{11}$ 
simultaneously, where $H_{00, 01}$ stands for the joint $H$ 
matrix striding quadrants $Q_{00}$ and $Q_{01}$,
and so on. Finally $H^* = H^*_{00, 01} H H^*_{10, 11}$, where 
$H$ and $H^*$ are over the entire region striding all four
quadrants. 
By \lemref{gap-squaring}, a square of size $2^k$ can 
thus be computed in $O(k)$ time, $O(2^{6k})$ work and space 
, and $O(2^{6k}/B)$ cache misses. 
Since a square's dimension at a higher 
level is geometrically larger than the
one at a lower level, the work, space, and cache bounds of 
the top-level $H^*$ computation then dominates. 
The conclusion then follows.
\end{proof}

\begin{theorem}
The final sublinear-time GAP algorithm in \cite{GalilPa94} takes
$O(\sqrt{n} \log n)$ time, $O(n^4)$ work and space 
, and $O(n^4/B)$ cache misses.
\label{thm:old-sublinear-gap}
\end{theorem}
\begin{proof}
Let $f(i, j)$ be the length of the shortest path from $(0, 0)$ to
$(i, j)$, their algorithm works as follows.
\begin{enumerate}
    \item It computes $H^*$ of the squares from bottom up 
        until level $k$ such
        that $2^k = v$, where $v$ is a parameter to be determined 
        later. Each square of $H^*(i, j, i+v, j+v)$ contains all
        pairs of shortest paths within the 2D region bounded
        between cell $(i, j)$ and $(i+v, j+v)$.
        By \lemref{gap-square-decomp}, each square's
        computation takes $O(\log^2 v)$ time, $O(v^6)$ work and 
        space, with $O(v^6/B)$ cache misses. Since there are 
        $O(n^2/v^2)$ of them, which can be computed simultaneously
        , the costs of this step sum up to $O(n^2 v^4)$ work 
        and space, with $O(n^2 v^4/B)$ cache misses.

    \item It computes $f(i, j)$ from cell 
        $(0, 0)$ to $(n, n)$ at step size of $v$ by 
        backward diagonal, i.e. $i + j$ is constant, in $2n/v$ 
        iterations.
        Squares on the same backward diagonal are computed 
        simultaneously in the same iteration.
        As in the 1D version (refer to the proof of 
        \thmref{old-sublinear-1D}), computing
        any square $f(i, j)$ on the $l$-th iteration, i.e.
        $i + j = lv$, requires squaring of following three 
        matrices.
        \begin{enumerate}
            \item All previous squares of 
                $f(0, 0) \twodots f(i-v, j-v)$, which can be 
                viewed as a combination of two independent 
                $1$-by-$(l-1)v$ 
                matrices. We denote it by a $(1 \times (l-1)v)^2$ 
                matrix for convenience.
            \item The $H$ values of $l$-th iteration, which
                can be viewed as a combination of two independent 
                $(l-1)v$-by-$v$ 
                matrices. We denote it by a $((l-1)v \times v)^2$
                matrix.
            \item The square of level-$k$ $H^*$ that is computed
                in above step (1) and is on the $l$-th 
                iteration, which can be viewed as a combination
                of two independent $v$-by-$v$ triangular matrices.
                We denote it by a $(v \times v)^2$ matrix.
        \end{enumerate}
        As in the 1D case, the squaring of
        $(1 \times (l-1)v)^2$ matrix with $((l-1)v \times v)^2$ 
        matrix dominates. By \lemref{gap-squaring}, the squaring 
        of one $f(i, j)$ takes $O(\log lv)$ time, $O(l^2 v^4)$ 
        work and space, with $O(l^2 v^4/B)$ cache misses.

        Summing over the $n/v$ iterations for $n^2/v^2$ 
        squares of $f$, the costs of this step are 
        $O(n/v \log n)$ time, $O(n^4)$ work and space, and 
        $O(n^4/B)$ cache misses.
\end{enumerate}

Summing up the overheads of the two steps and making 
$v = \sqrt{n}$ yields the bounds.
\end{proof}

\subsecput{cop-gap}{Cache-Oblivious Parallel GAP algorithm}
\input{co-gap-code}

Chowdhury and Ramachandran~\cite{ChowdhuryRa06, Chowdhury07} 
devised a cache-oblivious parallel (COP) algorithm for the 
GAP recurrences of \eqref{gap}. 
The algorithm separates
the updates to any quadrant to two functions, one is an update
from cells within the same quadrant, and the
other is an update from a disjoint quadrant either 
horizontally or vertically. The $\proc{cop-gap}_{\tri}$ function 
in \figref{cop-gap-tri} is the self-updating function,
the computational shape (total work) of which is a 
3D triangular analogue as shown in the upper-right corner of 
\figref{gap-work}. The $\proc{cop-gap-h}_{\Box}\allowbreak 
(A, B)$ function in \figref{cop-gap-box} is to update quadrant 
$A$ from a disjoint quadrant $B$ horizontally, the computational 
shape of which is a 3D cube as shown in the bottom-right corner 
of \figref{gap-work}. The update from a vertical direction is 
similar thus omitted.

\begin{lemma}
The COP algorithm in \figref{cop-gap-tri} computes the GAP
recurrences of \eqref{gap} in $O(n^{\log_2 3})$ time, 
optimal $O(n^2)$ space, optimal $O(n^3)$ work, and optimal 
$O(n^3 / (B \sqrt{M}))$ cache bounds 
\cite{ChowdhuryRa06, Chowdhury07}.
\label{lem:cop-gap}
\end{lemma}
%
% \punt{% begin punt
\begin{proof}
\paragrf{Space bound: }
Since different functions updating the same quadrant are 
explicitly separated by synchronizations, it's easy to see 
that it uses no more space than the input $2$D array of $A$. 

\paragrf{Time and cache bound: }
The algorithm has the following recurrences for 
time and cache bounds, which solves to $O(n^{\log_2 3})$ and
$O(n^3/(B\sqrt{M}))$, respectively. 
\punt{% begin punt
\begin{align*}
    T_{\infty, \tri}(n) &= 3 T_{\infty, \tri}(n/2) + 3 T_{\infty, \Box}(n/2) & T_{\infty, \Box}(n) &= 2 T_{\infty, \Box}(n/2) \\
    Q_{1, \tri} (n) &= 4 Q_{1, \tri} (n/2) + 4 Q_{1, \Box} (n/2) & Q_{1, \Box} (n) &= 8 Q_{1, \Box} (n/2) \\
    Q_{1, \tri} (n) &= Q_{1, \Box} (n) = O(n^2/B) & &\text{if $n^2 \leq \epsilon_8 M$} 
\end{align*}
}% end punt
\begin{align*}
    T_{\infty, \tri}(n) &= 3 T_{\infty, \tri}(n/2) + 3 T_{\infty, \Box}(n/2) \\
    T_{\infty, \Box}(n) &= 2 T_{\infty, \Box}(n/2) \\
    Q_{1, \tri} (n) &= 4 Q_{1, \tri} (n/2) + 4 Q_{1, \Box} (n/2) \\
    Q_{1, \Box} (n) &= 8 Q_{1, \Box} (n/2) \\
    Q_{1, \tri} (n) &= Q_{1, \Box} (n) = O(n^2/B) \quad \text{if $n^2 \leq \epsilon_8 M$} 
\end{align*}

\end{proof}
% }% end punt

\subsecput{nd-gap}{Nested Dataflow GAP algorithm}
\begin{theorem}
There is an ND algorithm that solves the GAP recurrences of 
\eqref{gap} in $O(n \log n)$ time, optimal $O(n^3)$ work, optimal 
$O(n^2)$ space, and optimal $O(n^3/(B\sqrt{M}))$ cache bounds. 
\label{thm:nd-gap}
\end{theorem}
\begin{proof}
Let's label all the quadrants recursively as in 
\figref{nd-gap-shape}.
\afigref{nd-gap-shape} is the 2D projection of a 3D 
triangular analogue in \figref{gap-work} from its 3D $o-xyz$ 
space to the $o-xy$ plane.
We have an observation that except the last $(11)$ 
quadrant, i.e. $Q_{11, 11, \ldots, 11}$, all $(11)$ 
quadrants nested in a recursion have the same amount
of computation as the $(00)$ quadrants of some of its parent's
siblings, i.e. the quadrants that lie on 
the same backward diagonal (i.e. $x + y$ is constant).
For instances, the $(11)$ quadrant of $Q_{00}$, i.e. 
$Q_{00, 11}$, has the same amount
of computation as the $(00)$ quadrants of some of its parent's 
siblings, i.e. $Q_{01, 00}$ and $Q_{10, 00}$;
$Q_{01, 11}$ and $Q_{10, 11}$ have the same
amount of computation as $Q_{11, 00}$, and so on.
Moreover, these quadrants on the same backward diagonal do
not have any data dependencies among each other so can
be scheduled to run in parallel.
By this observation, except the $Q_{11, 11, \ldots, 11}$
quadrant, all other $(11)$ quadrants can be pushed one
level up in the recursion and run simultaneously
with some of its parent's siblings as shown in  
\figreftwo{nd-gap-shape}{nd-gap-gamma-shape}. And this 
execution pattern proceeds recursively.

The pseudo-code of $\proc{nd-gap}_{\tri}$ algorithm is in 
\figref{nd-gap-tri}. 
The key improvement over the classic COP algorithm comes from 
associating every $\proc{cop-gap}_{\Box}$ subtask with its 
source $\proc{cop-gap}_{\tri}$ by a partially parallel ``$\fire$''
construct. The partially parallel $\fire$ 
construct will then be refined recursively by fire rule 
throughout computation. 
\punt{% begin punt
Since the refinement of a $\fire$ 
construct maintains only necessary data dependency 
to keep the algorithm correct, 
it won't introduce any artificial dependency 
\footnote{\defn{Artificial dependency} is extra control 
dependency in a parallel program than necessary 
data dependency. Artificial 
dependency is usually imposed by restricted parallel 
constructs \cite{TangYoKa15, DinhSiTa16}.} to hurt critical
path length as in classic cache-oblivious approach.
}% end punt
By fine-grain interleaving of data dependency across 
borders of recursion, each base case of 
$\proc{cop-gap}_{\Box}$ can start computing as soon as 
corresponding source subtask produces the data. 
At the same time, the recursive executing order among subtasks
is preserved. The interleaving just allows some subtasks at a 
higher level of recursion to start executing earlier.
The computation frontier proceeds by a 2D plane whose 
projection on the $o-xy$ plane aligns with some 
backward diagonal, i.e.  $x + y$ is constant. 
Referring to \figref{nd-gap-shape}, the 
projection of proceeding plane on the bottom 
$o-xy$ plane sweeps from cell $(0, 0)$ to $(n, n)$
by backward diagonals.
By contrast, a $\proc{cop-gap-h/v}_{\Box}$ computation in 
\figref{cop-gap-tri} has to wait until the entire 
$\proc{cop-gap}_{\tri}$ at the same recursion level finishes.
That is, it introduces excessive control dependency to subtasks 
at lower levels.
\input{nd-gap-code}

Referring to \figref{nd-gap-code}, the ND algorithm proceeds 
as follows.
\afigref{nd-gap-tri} says that the computation of a GAP recurrence 
without external dependency is accomplished by a partially 
parallel $\tfire{\Gamma\Box}$ computation followed by a recursion
on a geometrically smaller $Q_{11}$ quadrant. 
Note that the partially parallel 
$\tfire{\Gamma\Box}$ computation not only computes the 
$\Gamma$ shape region of $Q_{00, 01, 10}$ but also invokes 
$\proc{cop-gap}_{\Box}$ function (\figref{cop-gap-box}) 
to update $Q_{11}$ quadrant with the data from
the $\Gamma$ shape region. 
The big partially parallel $\tfire{\Gamma\Box}$ computation will 
then be refined by the fire rule in \figref{nd-gap-big-fire-rule} 
to two wavefronts that comprise 
four geometrically smaller $\tfire{\tri\Box}$ computations. The 
two wavefronts execute one after another as follows. 
Referring to \figref{nd-gap-gamma-shape}
\footnote{The $\oone{+}$ and $\oone{-}$
in \figreftwo{nd-gap-big-fire-rule}{nd-gap-small-fire-rule} are 
wildcards that will match at runtime to the source and sink 
subtasks of corresponding $\fire$ construct, respectively.
}
, it firstly executes the partially parallel computation of
$Q_{00} \tfire{\tri\Box}_{H/V} \{Q_{01}, Q_{10}\}$ 
(the first wavefront), 
\footnote{This is a shorthand for two partially parallel 
computations of $Q_{00} \tfire{\tri\Box}_H Q_{01}$ and $Q_{00} 
\tfire{\tri\Box}_V Q_{10}$.}
followed by $\{Q_{01}, Q_{10}\} \tfire{\tri\Box}_{H/V} Q_{11}$ 
(the second wavefront)
\footnote{Similarly, this is a shorthand for two partially 
parallel computations of $Q_{01} \tfire{\tri\Box}_V Q_{11}$ and 
$Q_{10} \tfire{\tri\Box}_H Q_{11}$}.
The partially parallel computations of the first wavefront share 
the same source and will have their sink tasks executed 
simultaneously.
By contrast, the second wavefront executes two distinct source 
tasks ($\proc{nd-gap}_{\tri}$) concurrently, while sink 
tasks ($\proc{cop-gap}_{\Box}$) serially because 
the two sinks write to the same output region of $Q_{11}$.
After the second wavefront, $Q_{11}$ quadrant will have been 
updated with the data from
the $\Gamma$ shape region of $Q_{00, 01, 10}$ and can start
a recursive GAP computation of $\proc{nd-gap}_{\tri}$ on 
itself without external data dependency.
The $\tfire{\tri\Box}$ construct recursively refines the 
partially parallel invocation of $\proc{nd-gap}_{\tri}$ 
($3$D triangular analogue) and $\proc{cop-gap}_{\Box}$
($3$D cube) to two parallel steps. There is a synchronization 
between the two parallel steps to guard correctness.
\afigref{nd-gap-small-fire-rule} shows the horizontal refinement 
, and the vertical direction is similar. 
Referring to 
\figreftwo{nd-gap-gamma-shape}{nd-gap-small-fire-rule},
$Q_{00} \tfire{\tri\Box} Q_{01}$ is refined to 
$Q_{00, 00} \tfire{\tri\Box} \{Q_{00, 01}, Q_{00, 10}, Q_{01, 00}
, Q_{01, 01}\}$ 
for the first parallel step and 
$\{Q_{00, 01}, Q_{00, 10}\} \tfire{\tri\Box} \{Q_{00, 00}
, Q_{00, 01}, Q_{00, 10}\}$ 
for the second parallel step. 
The rule of
$\{\oone{+}_{11} \tfire{\tri\Box}_H \{\oone{-}_{10}, 
\oone{-}_{11}\}\}$ in \figref{nd-gap-small-fire-rule}
is a nested $\tfire{\tri\Box}$ computation 
on a geometrically smaller $(11)$ quadrant, which will
run concurrently with some of its parent's siblings
on the same backward diagonal.

By the ND method, we have following recurrences for time 
bound.
\begin{align}
    & T_{\infty, \proc{nd-gap}_{\tri}} (n) = T_{\infty, \proc{nd-gap}_{\tfire{\Gamma\Box}}} (n) + T_{\infty, \proc{nd-gap}_{\tri}} (n/2) \label{eq:nd-gap-tri}\\
    & T_{\infty, \proc{nd-gap}_{\tfire{\Gamma\Box}}} (n) = T_{\infty, \proc{nd-gap}_{\tfire{\tri\Box}}} (n/2) \nonumber\\
    &           \quad\quad\quad\quad+ \max \{T_{\infty, \proc{nd-gap}_{\tri}}(n/4), T_{\infty, \proc{nd-gap}_{\tfire{\tri\Box}}} (n/2)\} \nonumber\\
    &           \quad\quad\quad\quad+ T_{\infty, \proc{cop-gap}_{\Box}} (n/2) \label{eq:nd-gap-big-fire}\\
    & T_{\infty, \proc{cop-gap}_{\Box}} (n) = 2 T_{\infty, \proc{cop-gap}_{\Box}} (n/2) \label{eq:cop-gap-box}\\
    & T_{\infty, \proc{nd-gap}_{\tfire{\tri\Box}}} (n) = 2 T_{\infty, \proc{nd-gap}_{\tfire{\tri\Box}}} (n/2) = O(n) \label{eq:nd-gap-small-fire}\\
    & T_{\infty, \proc{nd-gap}_{\tfire{\tri\Box}}} (1) = T_{\infty, \proc{nd-gap}_{\tri}} (1) + T_{\infty, \proc{nd-gap}_{\Box}} (1) \label{eq:nd-gap-small-fire-stop}
\end{align} 
The $\max\{\ldots\}$ term of \eqref{nd-gap-big-fire} says that 
the recursively nested $(11)$ quadrant will be executed 
concurrently with the second wavefront. 
Referring to \figref{nd-gap-gamma-shape}, these are the shaded 
$(11)$ quadrants. 
Since these $(11)$ quadrants are geometrically smaller, their 
execution will be completely subsumed 
in the $\max\{\ldots\}$, hence will not affect the critical-path
length. 
The last $T_{\infty, \proc{cop-gap}_{\Box}} (n/2)$ term of 
\eqref{nd-gap-big-fire} says that there are two partially 
parallel $\tfire{\tri\Box}$ computations of the second 
wavefront to update the same output quadrant ($Q_{11}$) so 
that one update has to lay behind.
Note that their source tasks have no dependency on 
each other, thus can still execute concurrently. 
\aeqref{cop-gap-box} recursively expands the computation of 
$\proc{cop-gap}_{\Box}$ without using any more space
than inputs and output.
\aeqref{nd-gap-small-fire-stop} says that the recursive 
refinement of $\tfire{\tri\Box}$ computation will stop and
reduce to classic serial construct (``$\serial$'') at 
base cases.
The time recurrences solve to 
$O(n \log n)$, which is asymptotically better than the classic 
COP algorithm in \lemref{cop-gap}.

Since the ND algorithm just schedules $\proc{cop-gap}_{\Box}$  
($3$D cubes) to run as soon as its sources 
($3$D triangular analogues) produce the data without changing
the recursive executing order, its work, space, and cache bounds 
do not change from those of \lemref{cop-gap}. 
\end{proof}

\punt{% begin punt -- Problem : how to explain the lower bound of LWS and GAP, 1D vs 2D?
Dynamic Programming is an approach of using storage for the intermediate
results to avoid re-computation of shared sub-problems, thus reducing
the total work.  The closure or matrix product method, on the contrary,
shortens the critical path by redundant computation.

\begin{corollary}
    Restricting to the method of Dynamic Programming, the ND GAP algorithm of 
    \thmref{nd-gap} is asymptotically optimal.
\end{corollary}
\begin{proof}
    Since the work, space, cache complexity of the ND GAP algorithm are already optimal, we
    consider only its time bound as follows.
    For the GAP problem of dimension $n$, there are $\Omega(n)$ elements on the critical 
    path, each of which takes $O(n)$ inputs and $\Omega(\log n)$ time to update on an 
    exclusive-write memory \cite{JaJa92}. The conclusion then follows.
\end{proof}
}% end punt

\subsecput{sub-nd-gap}{Work-Efficient and Sublinear-Time 
GAP Algorithm}
Referring to \figref{cop-gap-box}, we can see that the 
computation of $3$D cubes by $\proc{cop-gap}_\Box$ can be 
parallelized in the same way as that of \lemref{mm-n3-space}
by using temporary space. 
We then have following \lemref{sub-gap-box}.
\begin{lemma}
We have a sublinear-time algorithm to compute the $3$D cubes of 
dimension-$n$ in \figref{gap-work} in $O(\log n)$ time, $O(n^3)$ 
work and space, and $O(n^3/B)$ cache bounds.
\label{lem:sub-gap-box}
\end{lemma}

Combining Galil and Park's 
sublinear-time algorithm (\thmref{old-sublinear-gap}) with
the ND method of \thmref{nd-gap} for the 
$3$D triangular analogues, and \lemref{sub-gap-box}
for the $3$D cubes, we have the first work-efficient and
sublinear-time GAP algorithm as follows.
\begin{theorem}
    There is an ND algorithm that solves the GAP recurrences of
    \eqref{gap} in sublinear $O(n^{3/4} \log n)$ time, optimal
    $O(n^3)$ work, $O(n^3)$ space, and $O(n^3/B)$ cache bounds.  
\label{thm:sub-nd-gap}
\end{theorem}
\begin{proof}
Referring to \figref{gap-work}, we employ the algorithm 
of \lemref{sub-gap-box} to compute 
the $3$D cubes and the algorithm of \thmref{old-sublinear-gap} 
to compute the $3$D triangular analogues. 
At a high level, we preserve the ND scheduling as shown
by the recurrences of \eqrefTwo{nd-gap-tri}{nd-gap-small-fire} 
except that we stop the recursion of 
\eqref{nd-gap-small-fire} earlier than 
\eqref{nd-gap-small-fire-stop} at size $v$ as 
follows, where $v$ is a parameter to be determined later.
\begin{align}
& T_{\infty, \proc{nd-gap}_{\tfire{\tri\Box}}} (v) = T_{\infty, \proc{nd-gap}_{\tri}} (v) + T_{\infty, \proc{nd-gap}_{\Box}} (v) \label{eq:sub-nd-gap-small-fire-stop}
\end{align}
To bound the parameter $v$, we need to bound the total work of 
this hybrid algorithm to be the optimal $O(n^3)$. Since there are 
$O((n/v)^3)$ $3$D cubes, each of which takes $O(v^3)$ work 
(\lemref{sub-gap-box}), and $O((n/v)^2)$ $3$D triangular 
analogues, each of which takes $O(v^4)$ work 
(\thmref{old-sublinear-gap}), we have 
\eqref{gap-sub-nd-work} to bound the total work, which solves 
to $v = n^{1/2}$.
\begin{align}
    O((n/v)^3) \cdot O(v^3) + O((n/v)^2) \cdot O(v^{4}) &= O(n^3) \label{eq:gap-sub-nd-work}
\end{align}
Supplying this $v$ into \eqref{sub-nd-gap-small-fire-stop}
yields $T_{\infty, \proc{nd-gap}_{\tfire{\tri\Box}}} (v) = 
O(v^{1/2} \log v) = O(n^{1/4} \log n)$. Supplying this 
result into \eqrefTwo{nd-gap-tri}{nd-gap-small-fire} yields
the conclusion.
\end{proof}

\paragrf{Discussions :}
Our result of \thmref{nd-gap} improves over the prior 
cache-oblivious and cache-efficient algorithm of \lemref{cop-gap}
developed by Chowdhury and Ramachandran 
\cite{ChowdhuryRa06, Chowdhury07} on time bound without
sacrificing work, space, and cache efficiency.
By \thmref{sub-nd-gap}, we solve an open problem raised in 
Galil and Park's paper \cite{GalilPa94}, i.e. we have the
first work-efficient and sublinear-time GAP algorithm, though
its space and cache complexities are non-optimal, which will be 
left as yet another open problem for future work.

%% file: gap-algo-pic.tex
\newcommand{\len}{1}
\tikzset{
    bgap/.pic={
\coordinate (O) at (0, 0, 0);
\coordinate (A) at (\len, 0, 0);
\coordinate (B) at (\len, 0, \len);
\coordinate (C) at (0, 0, \len);
\coordinate (D) at (0, \len, 0);
\coordinate (E) at (\len, \len, 0);
\coordinate (F) at (\len, \len, \len);
\coordinate (G) at (0, \len, \len);
\coordinate (H) at (0, 2 * \len, 0);

% \draw[->] (O) -- (\len + 0.2, 0, 0) node [anchor=west] {$x$};
% \draw[->] (O) -- (0, \len + 0.2, 0) node [anchor=south]{$y$};
% \draw[->] (O) -- (0, 0, \len + 0.2) node [anchor=north east]{$z$};
\draw (C) -- (G) -- (B) -- cycle;
\draw (E) -- (A) -- (B) -- cycle;
\draw (G) -- (E) -- (B) -- cycle;
\draw (H) -- (G) -- (E) -- cycle;
\draw (H) -- (B);
\draw [dotted] (D) -- (O) -- (C);
\draw [dotted] (O) -- (A);
\draw [dotted] (H) -- (D) -- (G);
\draw [dotted] (D) -- (E);
}}
\tikzset{
    lgap/.pic={
\coordinate (O) at (0, 0, 0);
\coordinate (A) at (\len, 0, 0);
\coordinate (B) at (\len, 0, \len);
\coordinate (C) at (0, 0, \len);
\coordinate (D) at (0, \len, 0);
\coordinate (E) at (\len, \len, 0);
\coordinate (F) at (\len, \len, \len);
\coordinate (G) at (0, \len, \len);
\coordinate (H) at (0, 2 * \len, 0);

\draw (C) -- (G) -- (B) -- cycle;
\draw [dotted] (E) -- (A) -- (B) -- cycle;
\draw (G) -- (E) -- (B) -- cycle;
\draw (H) -- (G) -- (E) -- cycle;
\draw (H) -- (B);
\draw [dotted] (D) -- (O) -- (C);
\draw [dotted] (O) -- (A);
\draw [dotted] (H) -- (D) -- (G);
\draw [dotted] (D) -- (E);
}}
\tikzset{
    rgap/.pic={
\coordinate (O) at (0, 0, 0);
\coordinate (A) at (\len, 0, 0);
\coordinate (B) at (\len, 0, \len);
\coordinate (C) at (0, 0, \len);
\coordinate (D) at (0, \len, 0);
\coordinate (E) at (\len, \len, 0);
\coordinate (F) at (\len, \len, \len);
\coordinate (G) at (0, \len, \len);
\coordinate (H) at (0, 2 * \len, 0);

\draw (H) -- (G) -- (B) -- (E) -- cycle;
\draw (H) -- (B);
\draw (G) -- (E) -- (A);
\draw [dotted] (D) -- (E);
\draw [dotted] (O) -- (A);
}}
\tikzset{
    tgap/.pic={
\coordinate (O) at (0, 0, 0);
\coordinate (A) at (\len, 0, 0);
\coordinate (B) at (\len, 0, \len);
\coordinate (C) at (0, 0, \len);
\coordinate (D) at (0, \len, 0);
\coordinate (E) at (\len, \len, 0);
\coordinate (F) at (\len, \len, \len);
\coordinate (G) at (0, \len, \len);
\coordinate (H) at (0, 2 * \len, 0);

\draw (H) -- (G) -- (B) -- (E) -- cycle;
\draw (G) -- (E);
\draw (H) -- (B);
\draw [dotted] (O) -- (C) -- (G) -- (D) -- cycle;
\draw [dotted] (D) -- (E) -- (A) -- (O);
\draw [dotted] (H) -- (O);
\draw [dotted] (C) -- (B) -- (A);
}}
\tikzset{
    fbox/.pic={
\coordinate (O) at (0, 0, 0);
\coordinate (A) at (\len, 0, 0);
\coordinate (B) at (\len, 0, \len);
\coordinate (C) at (0, 0, \len);
\coordinate (D) at (0, \len, 0);
\coordinate (E) at (\len, \len, 0);
\coordinate (F) at (\len, \len, \len);
\coordinate (G) at (0, \len, \len);

\draw (C) -- (B) -- (F) -- (G) -- cycle;
\draw (F) -- (E) -- (A) -- (B);
\draw (G) -- (D) -- (E);
\draw [dotted] (D) -- (O) -- (C);
\draw [dotted] (O) -- (A);
}}
\tikzset{
    bbox/.pic={
\coordinate (O) at (0, 0, 0);
\coordinate (A) at (\len, 0, 0);
\coordinate (B) at (\len, 0, \len);
\coordinate (C) at (0, 0, \len);
\coordinate (D) at (0, \len, 0);
\coordinate (E) at (\len, \len, 0);
\coordinate (F) at (\len, \len, \len);
\coordinate (G) at (0, \len, \len);

\draw [dotted] (C) -- (B) -- (F) -- (G) -- cycle;
\draw [dotted] (F) -- (E) -- (A) -- (B);
\draw [dotted] (G) -- (D) -- (O) -- (C);
\draw [dotted] (D) -- (E);
\draw [dotted] (O) -- (A);
}}
\tikzset{
    lbox/.pic={
\coordinate (O) at (0, 0, 0);
\coordinate (A) at (\len, 0, 0);
\coordinate (B) at (\len, 0, \len);
\coordinate (C) at (0, 0, \len);
\coordinate (D) at (0, \len, 0);
\coordinate (E) at (\len, \len, 0);
\coordinate (F) at (\len, \len, \len);
\coordinate (G) at (0, \len, \len);

\draw (C) -- (B) -- (F) -- (G) -- cycle;
\draw [dotted] (F) -- (E) -- (A) -- (B);
\draw [dotted] (G) -- (D) -- (O) -- (C);
\draw [dotted] (D) -- (E);
\draw [dotted] (O) -- (A);
}}
\tikzset{
    rbox/.pic={
\coordinate (O) at (0, 0, 0);
\coordinate (A) at (\len, 0, 0);
\coordinate (B) at (\len, 0, \len);
\coordinate (C) at (0, 0, \len);
\coordinate (D) at (0, \len, 0);
\coordinate (E) at (\len, \len, 0);
\coordinate (F) at (\len, \len, \len);
\coordinate (G) at (0, \len, \len);

\draw (F) -- (E) -- (A) -- (B) -- cycle;
\draw [dotted] (O) -- (A);
}}

%\hspace*{-0.4cm}
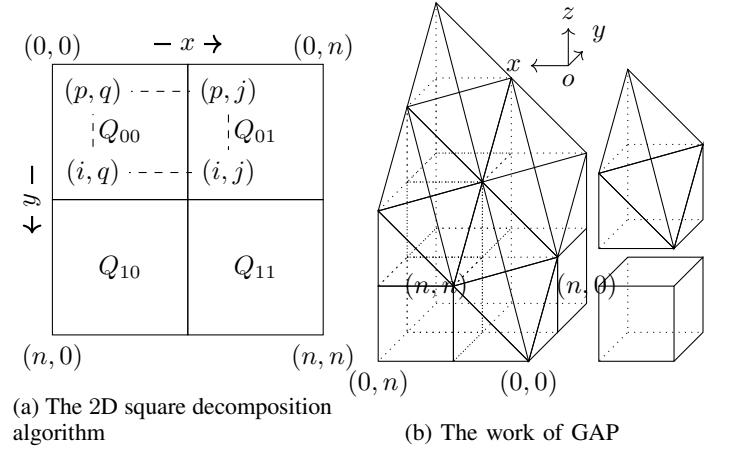
\begin{figure}[!ht]
\begin{subfigure}[b]{0.48 \linewidth}
\begin{center}
\begin{tikzpicture}[scale = 1.8]
\foreach \x in {0, 1} {
    \foreach \y in {0, 1} {
        \draw (\x, \y) rectangle (\x + 1, \y + 1);
    }
}
%\punt{ % begin punt
\draw [thick, ->] (0.75, 2.15) -- (1.25, 2.15) node[fill=white, pos=.5]{$x$}; 
\draw [thick, ->] (-0.15, 1.25) -- (-0.15, 0.75) node[fill=white, rotate=90, pos=.5]{$y$}; 
\node[above, yshift=-0.8mm] at (0, 2) {$(0, 0)$};
\node[above, yshift=-0.8mm] at (2, 2) {$(0, n)$};
\node[below] at (0, 0) {$(n, 0)$};
\node[below] at (2, 0) {$(n, n)$};

\coordinate (pq) at (0.3, 1.8);
\coordinate (pj) at (1.3, 1.8);
\coordinate (iq) at (0.3, 1.2);
\coordinate (ij) at (1.3, 1.2);

\draw [dashed] (pq) -- (pj) -- (ij);
\draw [dashed] (pq) -- (iq) -- (ij);

\path (pq) node [fill=white] {$(p, q)$} 
    -- (pj) node [fill=white] {$(p, j)$}
    -- (iq) node [fill=white] {$(i, q)$}
    -- (ij) node [fill=white] {$(i, j)$};
%} % end punt
\node at (0.5, 0.5) {$Q_{10}$};
\node at (1.5, 0.5) {$Q_{11}$};
\node at (0.5, 1.5) {$Q_{00}$};
\node at (1.5, 1.5) {$Q_{01}$};
\end{tikzpicture} 
\caption{The $2$D square decomposition algorithm}
\label{fig:gap-square-decomp}
\end{center}
\end{subfigure}
%\hfill
\begin{subfigure}[b]{0.52 \linewidth}
\begin{center}
\begin{tikzpicture}[scale = 1]

\node[above, yshift=-0.8mm] at (0, 0, 0) {$(n, n)$};
\node[above, yshift=-0.8mm] at (2, 0, 0) {$(n, 0)$};
\node[below] at (0, 0, 2) {$(0, n)$};
\node[below] at (2, 0, 2) {$(0, 0)$};

\begin{scope}[xshift=50, yshift=90]
%\node [anchor=east] at (0, 0, 0) {$O$};
\draw[->] (0, 0, 0) -- (-0.5, 0, 0) node [anchor=east] {$x$};
\draw[->] (0, 0, 0) -- (0, 0, -0.5) node [anchor=south west]{$y$};
\draw[->] (0, 0, 0) -- (0, 0.5, 0) node [anchor=south]{$z$};
\node[below] at (0, 0, 0) {$o$};
\end{scope}

\begin{scope}[]
\path (0, 0, 0) pic {bbox} (0, 1, 0) pic {bbox} (0, 2, 0) pic {tgap};
\path (1, 0, 0) pic {rbox} (1, 1, 0) pic {rgap};
\path (0, 0, 1) pic {lbox} (0, 1, 1) pic {lgap};
\path (1, 0, 1) pic {bgap};
\end{scope}

\begin{scope}[xshift=55]
\path (1, 0, 1) pic {fbox} (1, 1.5, 1) pic {bgap};
\end{scope}

\end{tikzpicture} 
\caption{The work of GAP}
\label{fig:gap-work}
\end{center}
\end{subfigure}
\caption{The GAP Problem}
\label{fig:gap}
\end{figure}

%% file: co-gap-code.tex
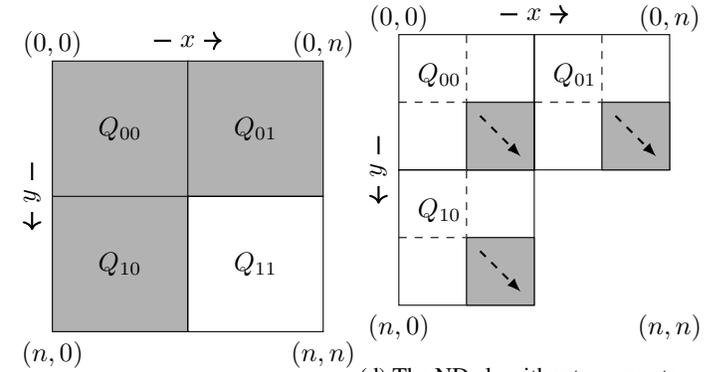
\begin{figure}[!ht]
\begin{subfigure}[b]{0.51 \linewidth}
\begin{center}
%\scriptsize
\begin{codebox}
\Procname{$\proc{cop-gap}_{\tri}(A)$}
\li $\proc{cop-gap}_{\tri}(A_{00})$
\li $\serial ((\proc{cop-gap-h}_{\Box}(A_{01}, A_{00})$
\zi \quad \quad $\serial \proc{cop-gap}_{\tri}(A_{01}))$
\li \quad $\parallel (\proc{cop-gap-v}_{\Box}(A_{10}, A_{00})$
\zi \quad \quad $\serial \proc{cop-gap}_{\tri}(A_{10})))$
\li $\serial \proc{cop-gap-v}_{\Box}(A_{11}, A_{01})$
\li $\serial \proc{cop-gap-h}_{\Box}(A_{11}, A_{10})$
\li $\serial \proc{cop-gap}_{\tri}(A_{11})$
\li \Return
\end{codebox}
\vspace*{-1 \baselineskip}
\caption{The COP algorithm for a triangular
region $A$.}
\label{fig:cop-gap-tri}
\end{center}
\end{subfigure}
\hfil
\begin{subfigure}[b]{0.48 \linewidth}
\begin{center}
%\scriptsize
\begin{codebox}
\Procname{$\proc{cop-gap-h}_{\Box}(A, B)$}
\li $(\proc{cop-gap-h}_{\Box}(A_{00}, B_{00})$
\li \Do $\serial \proc{cop-gap-h}_{\Box}(A_{00}, B_{01}))$
\End
\li $\parallel (\proc{cop-gap-h}_{\Box}(A_{01}, B_{00})$
\li \Do $\serial \proc{cop-gap-h}_{\Box}(A_{01}, B_{01}))$
\End
\li $\parallel (\proc{cop-gap-h}_{\Box}(A_{10}, B_{10})$
\li \Do $\serial \proc{cop-gap-h}_{\Box}(A_{10}, B_{11}))$
\End
\li $\parallel (\proc{cop-gap-h}_{\Box}(A_{11}, B_{10})$
\li \Do $\serial \proc{cop-gap-h}_{\Box}(A_{11}, B_{11}))$
\End
\li \Return
\end{codebox}
\vspace*{-1 \baselineskip}
\caption{The COP algorithm to update a
rectangular region $A$ from $B$.}
\label{fig:cop-gap-box}
\end{center}
\end{subfigure}
\vfil
\begin{subfigure}[b]{0.48 \linewidth}
\begin{center}
\input{nd-gap-algo}
\caption{The ND algorithm}
\label{fig:nd-gap-shape}
\end{center}
\end{subfigure}
\hfill
\begin{subfigure}[b]{0.48 \linewidth}
\begin{center}
\input{nd-gap-gamma}
\caption{The ND algorithm to compute a $\Gamma$-shape region.}
\label{fig:nd-gap-gamma-shape}
\end{center}
\end{subfigure}
\caption{Pseudo-code of cache-oblivious GAP algorithms}
\label{fig:co-gap}
\end{figure}

%% file: nd-gap-algo.tex
\begin{tikzpicture}[scale = 1.8]
\draw [thick, ->] (0.75, 2.15) -- (1.25, 2.15) node[fill=white, pos=.5]{$x$}; 
\draw [thick, ->] (-0.15, 1.25) -- (-0.15, 0.75) node[fill=white, rotate=90, pos=.5]{$y$}; 
\node[above, yshift=-0.8mm] at (0, 2) {$(0, 0)$};
\node[above, yshift=-0.8mm] at (2, 2) {$(0, n)$};
\node[below] at (0, 0) {$(n, 0)$};
\node[below] at (2, 0) {$(n, n)$};

\draw [fill=gray!60] (0, 0) rectangle (1, 1);
\node at (0.5, 0.5) {$Q_{10}$};

\draw [fill=gray!60] (0, 1) rectangle (1, 2);
\node at (0.5, 1.5) {$Q_{00}$};

\draw [fill=gray!60] (1, 1) rectangle (2, 2);
\node at (1.5, 1.5) {$Q_{01}$};

\draw (1, 0) rectangle (2, 1);
\node at (1.5, 0.5) {$Q_{11}$};
\end{tikzpicture}

%% file: nd-gap-gamma.tex
\begin{tikzpicture}[scale = 1.8]
\draw [thick, ->] (0.75, 2.15) -- (1.25, 2.15) node[fill=white, pos=.5]{$x$}; 
\draw [thick, ->] (-0.15, 1.25) -- (-0.15, 0.75) node[fill=white, rotate=90, pos=.5]{$y$}; 
\node[above, yshift=-0.8mm] at (0, 2) {$(0, 0)$};
\node[above, yshift=-0.8mm] at (2, 2) {$(0, n)$};
\node[below] at (0, 0) {$(n, 0)$};
\node[below] at (2, 0) {$(n, n)$};

\draw (0, 0) rectangle (1, 1);
\draw [dashed] (0, .5) -- (1, .5);
\draw [dashed] (.5, 0) -- (.5, 1);
\draw [fill=gray!60] (.5, 0) rectangle (1, .5);
\draw [-latex, thick, dashed] (.5 + .1, .5 - .1) -- (1 - .1, .1);
\node at (0.3, 0.7) {$Q_{10}$};

\draw (0, 1) rectangle (1, 2);
\draw [dashed] (0, 1.5) -- (1, 1.5);
\draw [dashed] (.5, 1) -- (.5, 2);
\draw [fill=gray!60] (.5, 1) rectangle (1, 1.5);
\draw [-latex, thick, dashed] (.5 + .1, 1 + .5 - .1) -- (1 - .1, 1+ .1);
\node at (0.3, 1.7) {$Q_{00}$};

\draw (1, 1) rectangle (2, 2);
\draw [dashed] (1, 1.5) -- (2, 1.5);
\draw [dashed] (1.5, 1) -- (1.5, 2);
\draw [fill=gray!60] (1.5, 1) rectangle (2, 1.5);
\draw [-latex, thick, dashed] (1 + .5 + .1, 1 + .5 - .1) -- (1 + 1 - .1, 1 + .1);
\node at (1.3, 1.7) {$Q_{01}$};
\end{tikzpicture}

%% file: nd-gap-code.tex
\begin{figure}[!ht]
\begin{minipage}[b]{0.48 \linewidth}
\begin{subfigure}[b]{ \linewidth}
\begin{center}
%\scriptsize
\begin{codebox}
\Procname{$\proc{nd-gap}_{\tri}(Q)$}
\li $\proc{nd-gap}_{\Gamma}(Q_{00, 01, 10})$
\zi \quad $\tfire{\Gamma\Box} \proc{cop-gap}_{\Box}(Q_{11}) \serial$
\li $\proc{nd-gap}_{\tri}(Q_{11}) \serial$
\li \Return
\end{codebox}
\vspace*{-1 \baselineskip}
\caption{The ND algorithm}
\label{fig:nd-gap-tri}
\end{center}
\end{subfigure}
\vfil
\begin{subfigure}[b]{ \linewidth}
\begin{center}
%\scriptsize
\begin{align*}
    & \oone{+} \tfire{\Gamma\Box} \oone{-} = \{\\ 
    &     \quad \oone{+}_{00} \tfire{\tri\Box}_{H/V} \{\oone{+}_{01}, \oone{+}_{10}\} \\
    &     \quad , \{\oone{+}_{01}, \oone{+}_{10}\} \tfire{\tri\Box}_{H/V} \{\oone{-}\}\}
\end{align*}
\vspace*{-1 \baselineskip}
\caption{The fire rule of ``$\tfire{\Gamma\Box}$'' computation}
\label{fig:nd-gap-big-fire-rule}
\end{center}
\end{subfigure}
\end{minipage}
\hfil
\begin{minipage}[b]{0.48 \linewidth}
%\vspace*{-1cm}
\begin{subfigure}[b]{ \linewidth}
\begin{center}
%\scriptsize
\begin{align*}
    &  \oone{+} \tfire{\tri\Box}_{H} \oone{-} = \{ \\
    & \quad \{\oone{+}_{00} \tfire{\tri\Box}_{H} \{\oone{+}_{01}, \oone{+}_{10} \\
    & \quad \quad , \oone{-}_{00}, \oone{-}_{01}\} \\
    & \quad , \oone{+}_{01} \tfire{\tri\Box}_{H} \{\oone{-}_{00}, \oone{-}_{01}\} \\
    & \quad , \oone{+}_{10} \tfire{\tri\Box}_{H} \{\oone{-}_{10}, \oone{-}_{11}\}\} \\
    & \quad , \{\oone{+}_{11} \tfire{\tri\Box}_{H} \{\oone{-}_{10}, \oone{-}_{11}\}\}\}
\end{align*}
\vspace*{-1 \baselineskip}
\caption{The fire rule of ``$\tfire{\tri\Box}$'' computation}
\label{fig:nd-gap-small-fire-rule}
\end{center}
\end{subfigure}
\end{minipage}
% \end{minipage}
\caption{Pseudo-code of ND GAP algorithm}
\label{fig:nd-gap-code}
\end{figure}

%% file: relWork.tex
\secput{concl}{Conclusion and Related Works}
\paragrf{Concluding Remarks: }
Galil and Park \cite{GalilPa94} gave out several work-efficient
and sublinear-time algorithms for DP recurrences with more than
$O(1)$ dependency. However their algorithm for the general GAP
problem is not work-efficient.
Classic COP approach \cite{ChowdhuryRa06, Chowdhury07, 
ChowdhuryLeRa10} usually attains optimal work, space, and
cache bounds in a cache-oblivious fashion, but with
a super-linear time bound due to both the updating order
imposed by DP recurrences and excessive control
dependency introduced by their approach.
In this paper, we present a new framework to parallelize a 
DP computation based on a novel combination of the closure 
method and ND method. 
Our results not only improves the time bounds of classic
cache-oblivious parallel algorithms without sacrificing
work, space, and cache efficiency, but also
gives out the first work-efficient and
sublinear-time algorithm for the general GAP problem, thus 
provides an answer to
the open problem raised by Galil and Park \cite{GalilPa94}.

\paragrf{Open Problem: }
It remains an interesting problem whether it is possible
to further bound the space and cache bounds of the general 
GAP algorithms
to be asymptotically optimal in a cache-oblivious fashion
while keeping its work bound to be optimal and time bound to
be sublinear.

\paragrf{Related Works: }
Galil, Giancarlo, and Park \cite{GalilGi89, GalilPa92, GalilPa94} 
proposed to solve DP
recurrences with more than $O(1)$ dependency by the methods of 
closure, matrix product, and indirection. 
% Their work is a great motivation for this paper. 
% \punt{% begin punt
Maleki et al. \cite{MalekiMuMy14} presented in their paper
that certain dynamic programming problem called ``Linear-Tropical
Dynamic Programming (LTDP)'' can possibly obtain extra parallelism
by breaking data dependencies between stages.
Their approach is based on the property of rank convergence 
of matrix multiplication in linear algebra.
\punt{% begin punt
Moreover, they multiply vector with matrix instead of matrix
with matrix to reduce total work. 
}% end punt
Their approach in the worst case can reduce to a serial algorithm.
% }% end punt
%
Chowdhury and Ramachandran \cite{ChowdhuryRa08} considered a
processor-aware hybrid $r$-way divide-and-conquer algorithms 
with different values of $r$ at different levels of recursion. 
Their cache bounds match asymptotically that of the best serial 
algorithm. 
Chowdhury et al. \cite{ChowdhurySiBl10} proposed a 
multicore-oblivious (MO) method for a hierarchical multi-level 
caching (HM) model. By the MO method, an algorithm requires
no specifying of any machine parameters such as the number of
computing cores, number of cache levels, cache size, or block
transfer size. However, for improved performance, an MO algorithm
is allowed to provide advices or ``hints'' to the runtime scheduler
through a small set of instructions on how to schedule the
parallel tasks it spawns.
%
% 3. Automatic generation of Strassen-like algorithm, which 
% doesn't help the space bound
\punt{% begin punt
General matrix multiplication can be viewed as a special case
of DP recurrence with $O(n)$ dependency.
There are wide studies on the problem or its Strassen-like fast
variants with various tradeoffs among work, 
space, time, and communication bounds \cite{BensonBa15,
BallardDeHo12, McCollTi99, SolomonikDe11,
KumarHuSa95, BoyerDuPe09, HuangSmHe16, SmithGeSm14}.
The basic idea of these works is to switch back and forth between
a serial algorithm to save and reuse space and a parallel
algorithm to increase parallelism. 
}% end punt
Shun et al. \cite{ShunBlFi13} proposed ``priority updates'' to
relax the serialization of ``concurrent writes'' to the same
memory cell, thus could possibly improve the time bounds in
some cases. However, not all operations can be prioritized.

% We may need more related works.